\documentclass{amsart}
\usepackage{amsfonts}
\usepackage{amsmath,amscd}
\usepackage{amsthm}
\usepackage{amssymb}
\usepackage{latexsym}
\newtheorem{lem}{Lemma}[section]

\newtheorem{note}{Note}

\newtheorem{example}{Example}
\newtheorem{defn}{Definition}[section]
\newtheorem{thm}{Theorem}
\newtheorem{conj}{Conjecture}
\newtheorem{rem}{Remark}

\numberwithin{equation}{section}
\begin{document}
\newcommand{\beqa}{\begin{eqnarray}}
\newcommand{\eeqa}{\end{eqnarray}}
\newcommand{\thmref}[1]{Theorem~\ref{#1}}
\newcommand{\secref}[1]{Sect.~\ref{#1}}
\newcommand{\lemref}[1]{Lemma~\ref{#1}}
\newcommand{\propref}[1]{Proposition~\ref{#1}}
\newcommand{\corref}[1]{Corollary~\ref{#1}}
\newcommand{\remref}[1]{Remark~\ref{#1}}
\newcommand{\er}[1]{(\ref{#1})}
\newcommand{\nc}{\newcommand}
\newcommand{\rnc}{\renewcommand}

\nc{\cal}{\mathcal}

\nc{\goth}{\mathfrak}
\rnc{\bold}{\mathbf}
\renewcommand{\frak}{\mathfrak}
\renewcommand{\Bbb}{\mathbb}

\newcommand{\id}{\text{id}}

\nc\K{\mathbb K}
\nc{\Cal}{\mathcal}
\nc{\Xp}[1]{X^+(#1)}
\nc{\Xm}[1]{X^-(#1)}
\nc{\on}{\operatorname}
\nc{\ch}{\mbox{ch}}
\nc{\Z}{{\bold Z}}
\nc{\J}{{\mathcal J}}
\nc{\C}{{\bold C}}
\nc{\Q}{{\bold Q}}
\renewcommand{\P}{{\mathcal P}}
\nc{\N}{{\Bbb N}}
\nc\beq{\begin{equation}}
\nc\enq{\end{equation}}
\nc\lan{\langle}
\nc\ran{\rangle}
\nc\bsl{\backslash}
\nc\mto{\mapsto}
\nc\lra{\leftrightarrow}
\nc\hra{\hookrightarrow}
\nc\sm{\smallmatrix}
\nc\esm{\endsmallmatrix}
\nc\sub{\subset}
\nc\ti{\tilde}
\nc\nl{\newline}
\nc\fra{\frac}
\nc\und{\underline}
\nc\ov{\overline}
\nc\ot{\otimes}
\nc\bbq{\bar{\bq}_l}
\nc\bcc{\thickfracwithdelims[]\thickness0}
\nc\ad{\text{\rm ad}}
\nc\Ad{\text{\rm Ad}}
\nc\Hom{\text{\rm Hom}}
\nc\End{\text{\rm End}}
\nc\Ind{\text{\rm Ind}}
\nc\Res{\text{\rm Res}}
\nc\Ker{\text{\rm Ker}}
\rnc\Im{\text{Im}}
\nc\sgn{\text{\rm sgn}}
\nc\tr{\text{\rm tr}}
\nc\Tr{\text{\rm Tr}}
\nc\supp{\text{\rm supp}}
\nc\card{\text{\rm card}}
\nc\bst{{}^\bigstar\!}
\nc\he{\heartsuit}
\nc\clu{\clubsuit}
\nc\spa{\spadesuit}
\nc\di{\diamond}
\nc\cW{\cal W}
\nc\cG{\cal G}
\nc\al{\alpha}
\nc\bet{\beta}
\nc\ga{\gamma}
\nc\de{\delta}
\nc\ep{\epsilon}
\nc\io{\iota}
\nc\om{\omega}
\nc\si{\sigma}
\rnc\th{\theta}
\nc\ka{\kappa}
\nc\la{\lambda}
\nc\ze{\zeta}

\nc\vp{\varpi}
\nc\vt{\vartheta}
\nc\vr{\varrho}

\nc\Ga{\Gamma}
\nc\De{\Delta}
\nc\Om{\Omega}
\nc\Si{\Sigma}
\nc\Th{\Theta}
\nc\La{\Lambda}

\nc\boa{\bold a}
\nc\bob{\bold b}
\nc\boc{\bold c}
\nc\bod{\bold d}
\nc\boe{\bold e}
\nc\bof{\bold f}
\nc\bog{\bold g}
\nc\boh{\bold h}
\nc\boi{\bold i}
\nc\boj{\bold j}
\nc\bok{\bold k}
\nc\bol{\bold l}
\nc\bom{\bold m}
\nc\bon{\bold n}
\nc\boo{\bold o}
\nc\bop{\bold p}
\nc\boq{\bold q}
\nc\bor{\bold r}
\nc\bos{\bold s}
\nc\bou{\bold u}
\nc\bov{\bold v}
\nc\bow{\bold w}
\nc\boz{\bold z}

\nc\ba{\bold A}
\nc\bb{\bold B}
\nc\bc{\bold C}
\nc\bd{\bold D}
\nc\be{\bold E}
\nc\bg{\bold G}
\nc\bh{\bold H}
\nc\bi{\bold I}
\nc\bj{\bold J}
\nc\bk{\bold K}
\nc\bl{\bold L}
\nc\bm{\bold M}
\nc\bn{\bold N}
\nc\bo{\bold O}
\nc\bp{\bold P}
\nc\bq{\bold Q}
\nc\br{\bold R}
\nc\bs{\bold S}
\nc\bt{\bold T}
\nc\bu{\bold U}
\nc\bv{\bold V}
\nc\bw{\bold W}
\nc\bz{\bold Z}
\nc\bx{\bold X}

\nc\ca{\mathcal A}
\nc\cb{\mathcal B}
\nc\cc{\mathcal C}
\nc\cd{\mathcal D}
\nc\ce{\mathcal E}
\nc\cf{\mathcal F}
\nc\cg{\mathcal G}
\rnc\ch{\mathcal H}
\nc\ci{\mathcal I}
\nc\cj{\mathcal J}
\nc\ck{\mathcal K}
\nc\cl{\mathcal L}
\nc\cm{\mathcal M}
\nc\cn{\mathcal N}
\nc\co{\mathcal O}
\nc\cp{\mathcal P}
\nc\cq{\mathcal Q}
\nc\car{\mathcal R}
\nc\cs{\mathcal S}
\nc\ct{\mathcal T}
\nc\cu{\mathcal U}
\nc\cv{\mathcal V}
\nc\cz{\mathcal Z}
\nc\cx{\mathcal X}
\nc\cy{\mathcal Y}

\nc\e[1]{E_{#1}}
\nc\ei[1]{E_{\delta - \alpha_{#1}}}
\nc\esi[1]{E_{s \delta - \alpha_{#1}}}
\nc\eri[1]{E_{r \delta - \alpha_{#1}}}
\nc\ed[2][]{E_{#1 \delta,#2}}
\nc\ekd[1]{E_{k \delta,#1}}
\nc\emd[1]{E_{m \delta,#1}}
\nc\erd[1]{E_{r \delta,#1}}

\nc\ef[1]{F_{#1}}
\nc\efi[1]{F_{\delta - \alpha_{#1}}}
\nc\efsi[1]{F_{s \delta - \alpha_{#1}}}
\nc\efri[1]{F_{r \delta - \alpha_{#1}}}
\nc\efd[2][]{F_{#1 \delta,#2}}
\nc\efkd[1]{F_{k \delta,#1}}
\nc\efmd[1]{F_{m \delta,#1}}
\nc\efrd[1]{F_{r \delta,#1}}

\nc\fa{\frak a}
\nc\fb{\frak b}
\nc\fc{\frak c}
\nc\fd{\frak d}
\nc\fe{\frak e}
\nc\ff{\frak f}
\nc\fg{\frak g}
\nc\fh{\frak h}
\nc\fj{\frak j}
\nc\fk{\frak k}
\nc\fl{\frak l}
\nc\fm{\frak m}
\nc\fn{\frak n}
\nc\fo{\frak o}
\nc\fp{\frak p}
\nc\fq{\frak q}
\nc\fr{\frak r}
\nc\fs{\frak s}
\nc\ft{\frak t}
\nc\fu{\frak u}
\nc\fv{\frak v}
\nc\fz{\frak z}
\nc\fx{\frak x}
\nc\fy{\frak y}

\nc\fA{\frak A}
\nc\fB{\frak B}
\nc\fC{\frak C}
\nc\fD{\frak D}
\nc\fE{\frak E}
\nc\fF{\frak F}
\nc\fG{\frak G}
\nc\fH{\frak H}
\nc\fJ{\frak J}
\nc\fK{\frak K}
\nc\fL{\frak L}
\nc\fM{\frak M}
\nc\fN{\frak N}
\nc\fO{\frak O}
\nc\fP{\frak P}
\nc\fQ{\frak Q}
\nc\fR{\frak R}
\nc\fS{\frak S}
\nc\fT{\frak T}
\nc\fU{\frak U}
\nc\fV{\frak V}
\nc\fZ{\frak Z}
\nc\fX{\frak X}
\nc\fY{\frak Y}
\nc\tfi{\ti{\Phi}}
\nc\bF{\bold F}
\rnc\bol{\bold 1}

\nc\ua{\bold U_\A}

\nc\qinti[1]{[#1]_i}
\nc\q[1]{[#1]_q}
\nc\xpm[2]{E_{#2 \delta \pm \alpha_#1}}  
\nc\xmp[2]{E_{#2 \delta \mp \alpha_#1}}
\nc\xp[2]{E_{#2 \delta + \alpha_{#1}}}
\nc\xm[2]{E_{#2 \delta - \alpha_{#1}}}
\nc\hik{\ed{k}{i}}
\nc\hjl{\ed{l}{j}}
\nc\qcoeff[3]{\left[ \begin{smallmatrix} {#1}& \\ {#2}& \end{smallmatrix}
\negthickspace \right]_{#3}}
\nc\qi{q}
\nc\qj{q}

\nc\ufdm{{_\ca\bu}_{\rm fd}^{\le 0}}


\nc\isom{\cong} 

\nc{\pone}{{\Bbb C}{\Bbb P}^1}
\nc{\pa}{\partial}
\def\H{\mathcal H}
\def\L{\mathcal L}
\nc{\F}{{\mathcal F}}
\nc{\Sym}{{\goth S}}
\nc{\A}{{\mathcal A}}
\nc{\arr}{\rightarrow}
\nc{\larr}{\longrightarrow}

\nc{\ri}{\rangle}
\nc{\lef}{\langle}
\nc{\W}{{\mathcal W}}
\nc{\uqatwoatone}{{U_{q,1}}(\su)}
\nc{\uqtwo}{U_q(\goth{sl}_2)}
\nc{\dij}{\delta_{ij}}
\nc{\divei}{E_{\alpha_i}^{(n)}}
\nc{\divfi}{F_{\alpha_i}^{(n)}}
\nc{\Lzero}{\Lambda_0}
\nc{\Lone}{\Lambda_1}
\nc{\ve}{\varepsilon}
\nc{\phioneminusi}{\Phi^{(1-i,i)}}
\nc{\phioneminusistar}{\Phi^{* (1-i,i)}}
\nc{\phii}{\Phi^{(i,1-i)}}
\nc{\Li}{\Lambda_i}
\nc{\Loneminusi}{\Lambda_{1-i}}
\nc{\vtimesz}{v_\ve \otimes z^m}

\nc{\asltwo}{\widehat{\goth{sl}_2}}
\nc\ag{\widehat{\goth{g}}}  
\nc\teb{\tilde E_\boc}
\nc\tebp{\tilde E_{\boc'}}

\newcommand{\LR}{\bar{R}}
\newcommand{\eeq}{\end{equation}}
\newcommand{\ben}{\begin{eqnarray}}
\newcommand{\een}{\end{eqnarray}}

\title[Higher order relations for the $q-$Onsager algebra]{Analogues of Lusztig's higher order relations\\ for the $q-$Onsager algebra}
\author{P. Baseilhac and T. T. Vu}
\address{Laboratoire de Math\'ematiques et Physique Th\'eorique CNRS/UMR 7350,
        F\'ed\'eration Denis Poisson, Universit\'e de Tours, Parc de Grammont, 37200 Tours, FRANCE}
\email{baseilha@lmpt.univ-tours.fr; Thi-thao.Vu@lmpt.univ-tours.fr}

\begin{abstract} Let $A,A^*$ be the generators of the $q-$Onsager algebra. Analogues of Lusztig's $r-th$ higher order relations
are proposed. In a first part, based on the properties of tridiagonal pairs of $q-$Racah type which satisfy the defining 
relations of the $q-$Onsager algebra, higher order relations are derived for $r$ generic. The coefficients entering in the relations 
are determined from a two-variable polynomial generating function. In a second part, it is conjectured that $A,A^*$ 
satisfy the higher order relations previously obtained. The conjecture is proven for $r=2,3$. For $r$ generic, 
using an inductive argument recursive formulae for the coefficients are derived. The conjecture is checked for several
values of $r\geq 4$. Consequences for coideal subalgebras and integrable systems with  boundaries at $q$ a root of unity are pointed out.
\end{abstract}

\maketitle

\vskip -0.2cm



\vspace{5mm}

\section{Introduction}
Consider the quantum universal enveloping algebras for arbitrary Kac-Moody algebras $\widehat{g}$ introduced by Drinfeld \cite{Dr} and Jimbo \cite{Jim}.  Let  $\{a_{ij}\}$ be the extended Cartan matrix. Fix coprime integers $d_i$ such that $d_ia_{ij}$ is symmetric. Define $q_i=q^{d_i}$. The quantum universal enveloping algebra $U_q(\widehat{g})$   is generated by the elements $\{h_j,e_j,f_j\}$, $j=0,1,...,rank(g)$, which satisfy the defining relations:
\beqa
 &&[h_i,h_j]=0\ , \quad [h_i,e_j]=a_{ij}e_j\ , \quad
[h_i,f_j]=-a_{ij}f_j\ ,\quad
[e_i,f_j]=\delta_{ij}\frac{q_i^{h_i}-q_i^{-h_i}}{q_i-q_i^{-1}},  \quad i,j=0,1,...,rank(g) \ , \ \nonumber
\nonumber\eeqa
together with the so-called {\it quantum Serre relations}\footnote{As usual, we denote: $
\left[ \begin{array}{c}
n \\
m 
\end{array}\right]_q
=\frac{[n]_q!}{[m]_q!\,[n-m]_q!}\ , \qquad
[n]_q!=\prod_{l=1}^n[l]_q\ ,\qquad
[n]_q=\frac{q^n-q^{-n}}{q-q^{-1}}, \quad [0]_q=1 \ .$} 
($i\neq j$)
\beqa
&& \sum_{k=0}^{1-a_{ij}}(-1)^k
\left[ \begin{array}{c}
1-a_{ij} \\
k 
\end{array}\right]_{q_i}
e_i^{1-a_{ij}-k}\,e_j\,e_i^k=0\ , \qquad 
\sum_{k=0}^{1-a_{ij}}(-1)^k
\left[ \begin{array}{c}
1-a_{ij} \\
k
\end{array}\right]_{q_i}
f_i^{1-a_{ij}-k}\,f_j\,\,f_i^k=0\ .  \label{Uqserre}
\eeqa
\vspace{1mm}

In the mathematical literature \cite{Luzt}, generalizations of the relations (\ref{Uqserre}) - the so-called {\it higher order quantum ($q-$)Serre relations} - have been proposed. For  $\widehat{g}=\widehat{sl_2}$, they read\footnote{For  $\widehat{g}=\widehat{sl_2}$, recall that $a_{ii}=2$, $a_{ij}=-2$  with $i,j=0,1$.}  \cite{Luzt}:
\beqa
\sum_{k=0}^{2r+1} (-1)^{k} \,  \left[ \begin{array}{c}
2r+1 \\
k
\end{array}\right]_{q} e_i^{2 r+1-k} e_j^{r} e_i^{k}&=&0\ , \label{hqSerre}\\
\sum_{k=0}^{2r+1} (-1)^{k}  \,  \left[ \begin{array}{c}
2r+1 \\
k
\end{array}\right]_{q} f_i^{2r+1-k} f_j^{r} f_i^{k}&=&0 \ \ \quad \mbox{for} \quad i\neq j,\ \  i,j=0,1 \ . \nonumber
\eeqa
Note that the higher order $q-$Serre relations (\ref{hqSerre}) appear in different contexts. On one hand, they determine the commutation relations among the elements of the $U_q(\widehat{g})$-subalgebra generated by the divided powers of $e_i,f_i$ which can be also obtained  using the braid group action of $U_q(\widehat{g})$. They also arise in the discussion of  the quantum Frobenius homomorphism \cite{Luzt}. On the other hand, as will be mentioned in the last Section, in the integrable systems literature the higher order $q-$Serre relations (\ref{hqSerre}) play a central role in the identification of the symmetry of quantum integrable models at $q$ a root of unity \cite{DFM,KM,AYP,ND}. 
 For instance, using (\ref{hqSerre}) one shows that the XXZ spin chain with periodic boundary conditions enjoys a $\widehat{sl_2}$ loop symmetry at $q$ a root of unity \cite{DFM}.Also, they are essential to derive the Serre relations for the basic generators of the superintegrable chiral Potts model \cite{ND,AYP}. \vspace{1mm}

In recent years, a new algebraic structure called the $q-$Onsager algebra that is closely related with $U_q(\widehat{sl_2})$ has appeared in the mathematical \cite{Ter03} and integrable systems  \cite{B1} literature. In particular, a  relation between the $q-$Onsager algebra, the spectral parameter dependent reflection equation \cite{Cher84,Skly88} and coideal subalgebras of $U_q(\widehat{sl_2})$ has been exhibited \cite{IT32,B1,BK1,BasS,BB0,Kolb}.  Recall that the tridiagonal algebra introduced in the context of $P-$ and $Q-$polynomial association schemes \cite{Ter03} is an associative algebra with unit which consists of two generators $ A$ and $A^*$ called the standard generators. In general, the defining relations depend on five scalars $\rho,\rho^*,\gamma,\gamma^*$ and $\beta$. The $q-$Onsager algebra is a special case of the tridiagonal algebra: it corresponds to the {\it reduced} parameter sequence $\gamma=0,\gamma^*=0$, $\beta=q^2+q^{-2}$ and $\rho=\rho_0$, $\rho^*=\rho_1$  which exhibits all interesting properties that can be extended to more general parameter sequences. The defining relations of the $q-$Onsager algebra read\footnote{The $q-$commutator $\big[X,Y\big]_q=qXY-q^{-1}YX$ is introduced, where $q$ is called the deformation parameter.}
\beqa
[A,[A,[A,A^*]_q]_{q^{-1}}]=\rho_0[A,A^*]\
,\qquad
[A^*,[A^*,[A^*,A]_q]_{q^{-1}}]=\rho_1[A^*,A]\
\label{qDG} , \eeqa
which can be seen as $\rho_i-$deformed analogues of the $q-$Serre relations (\ref{Uqserre}) associated with $\widehat{g}\equiv \widehat{sl_2}$. For $q=1$, $\rho_0=\rho_1=16$, note that they coincide with the Dolan-Grady relations \cite{DG}. 
\vspace{1mm}

 In the study of tridiagonal algebras and the representation theory associated with the special case $\rho_0=\rho_1=0$,  higher order $q-$Serre relations (\ref{hqSerre}) play an important role \cite{IT03} in the construction of a basis of the corresponding vector space. As  suggested in \cite[Problem~3.4]{IT03}, for $\rho_0\neq 0$, $\rho_1\neq 0$ finding analogues of the higher order $q-$Serre relations for the $q-$Onsager algebra is an interesting problem.  Another interest for the construction of higher order relations  associated with the $q-$Onsager algebra comes from the theory of quantum integrable systems with boundaries. Indeed, by analogy with the case of periodic boundary conditions \cite{DFM}, such relations should play a central role in the identification of the symmetry of the Hamiltonian of the XXZ open spin chain at $q$ a root of unity and special boundary parameters. \vspace{1mm}

Motivated by these open problems, in the present paper we propose the $r-th$ higher order tridiagonal relations associated with the $q-$Onsager algebra (\ref{qDG}), which we refer as the {\it $r-th$  higher order $q-$Dolan-Grady relations}. As will be argued, they can be written in the form:
\beqa
\sum_{p=0}^{r}\sum_{j=0}^{2r+1-2p}  (-1)^{j+p}  \rho_0^{p}\, {c}_{j}^{[r,p]}\,  A^{2 r+1-2p-j} {A^*}^{r} A^{j}&=&0\  \label{qDGfinr} \ ,\\
\sum_{p=0}^{r}\sum_{j=0}^{2r+1-2p} (-1)^{j+p} \rho_1^{p} \, {c}_{j}^{[r,p]}\,  {A^*}^{2 r+1-2p-j} {A}^{r} {A^*}^{j}&=&0\  
\ \nonumber
\eeqa
where $c_{2(r-p)+1-j}^{[r,p]} =c_{j}^{[r,p]}$ and\footnote{Here $\{ x\}$ denotes the integer part of $x$. Let $j,m,n$ be integers, we write $j=\overline{m,n}$  for $j=m,m+1,...,n-1,n$.}
\beqa
\qquad &&c_{j}^{[r,p]} =  \sum_{k=0}^j   \frac{(r-p)!}{(\{\frac{j-k}{2}\})!(r-p-\{\frac{j-k}{2}\})!}
    \sum_{{\cal P}}  [s_1]^2_{q^2}...[s_p]^2_{q^2}  \frac{[2s_{p+1}]_{q^2}...[2s_{p+k}]_{q^2}}{[s_{p+1}]_{q^2}...[s_{p+k}]_{q^2}}  \label{cfinr}
\eeqa
\beqa 
\mbox{with}\quad \ \left\{\begin{array}{cc}
\!\!\! \!\!\! \!\!\! \!\!\!  \!\!\! \!\!\! \!\!\! \!\!\!   \!\!\! \!\!\! \!\!\! \!\!\! j= \overline{0,r-p}\ , \quad s_i\in\{1,2,...,r\}\ ,\\
{\cal P}: \begin{array}{cc} \ \ s_1<\dots<s_p\ ;\quad \ s_{p+1}<\dots<s_{p+k}\ ,\\
 \{s_{1},\dots,s_{p}\} \cap \{s_{p+1},\dots,s_{p+k}\}=\emptyset \end{array}
\end{array}\right.\ .\nonumber
\eeqa
For $\rho_0=\rho_1=0$, the coefficients reduce to the ones in  (\ref{hqSerre}) \cite{Luzt}.\vspace{1mm}
\vspace{1mm}

The paper is organized as follows. In the next Section, we briefly recall the notion of tridiagonal (TD) pairs and tridiagonal systems (see \cite{TD10} and references therein). Based on the properties of TD pairs of $q-$Racah type \cite{NT:muqrac}, higher order tridiagonal relations satisfied by TD pairs are constructed (cf. Theorem \ref{hqdg}). For the reduced parameter sequence (see above comments), it is known \cite{Ter03} that the basic tridiagonal relations reduce to (\ref{qDG}). In this case, it is shown that the higher order tridiagonal relations associated with (\ref{qDG}) take the form (\ref{qDGfinr}) for $r$ generic. See Example \ref{exfinr}. A two-variable polynomial generating function  for the coefficients is proposed, which gives (\ref{cfinr}). Motivated by these results, in Section 3 it is more generally  conjectured that given $A,A^*$ satisfying (\ref{qDG}) the $r-th$ higher order $q-$Dolan-Grady relations (\ref{qDGfinr}) with (\ref{cfinr}) are satisfied (see conjecture \ref{conj1}). First, we prove the conjecture in the simplest examples for $r=2,3$. Then, by induction the existence and structure of relations of the form (\ref{qDGfinr}) for generic values of $r$ are studied, leading to explicit recursive formulae for the coefficients ${c}_{k}^{[r,p]}$.  For several values of $r\geq 4$, the coefficients obtained by both approaches - from the properties of TD pairs or by induction -  are found to coincide, giving another support for the proposal. In the last Section,  potential applications of the higher $q-$Dolan-Grady relations (\ref{qDGfinr}) to the theory of coideal subalgebras and integrable systems with boundaries at $q$ a root of unity are pointed out. In Appendices A,B useful recursion relations are reported. 

\vspace{2mm}

{\bf Notations:} Throughout this paper  $\K$ denotes a field. $q$ is assumed not to be a root of unity.
\vspace{1mm}

\vspace{3mm}
\section{Higher order relations from the theory of tridiagonal pairs}
The main result of this Section is Theorem \ref{hqdg}, which follows from previous works on tridiagonal pairs. Note that the basic material (Definitions 2.1-2.2, Lemma 2.2 and Theorem 1) introduced in the beginning of this Section is essentially taken from \cite{TD00,TD10,NT:muqrac}. As an application of Theorem 2, the higher order $q-$Dolan-Grady relations (\ref{qDGfinr}) are derived  in Example 2 for $r$ generic.

\subsection{Tridiagonal pairs of $q-$Racah type}
Let $V$ denote a vector space over $\K$ with finite
positive dimension. For a  linear transformation $A:V\to V$
and a subspace $W \subseteq V$, we call $W$ an
 {\it eigenspace} of $A$ whenever $W\not=0$ and there exists $\theta \in \K$ such that 
$W=\lbrace v \in V \;\vert \;Av = \theta v\rbrace$, $\theta$ is the {\it eigenvalue} of
$A$ associated with $W$. $A$ is {\it diagonalizable} whenever $V$ is spanned by the eigenspaces of $A$.

\begin{defn}  
{\rm \cite[Definition~1.1]{TD00}}
\label{def:tdp}
\rm
Let $V$ denote a vector space over $\K$ with finite
positive dimension. 
By a {\it tridiagonal pair} (or {\it $TD$ pair})
on $V$
we mean an ordered pair of linear transformations
$A:V \to V$ and 
$A^*:V \to V$ 
that satisfy the following four conditions.
\begin{itemize}
\item[(i)] Each of $A,A^*$ is diagonalizable.
\item[(ii)] There exists an ordering $\lbrace V_i\rbrace_{i=0}^d$ of the  
eigenspaces of $A$ such that 
\begin{equation}
A^* V_i \subseteq V_{i-1} + V_i+ V_{i+1} \qquad \qquad 0 \leq i \leq d,
\label{eq:t1}
\end{equation}
where $V_{-1} = 0$ and $V_{d+1}= 0$.
\item[(iii)] There exists an ordering $\lbrace V^*_i\rbrace_{i=0}^{\delta}$ of
the  
eigenspaces of $A^*$ such that 
\begin{equation}
A V^*_i \subseteq V^*_{i-1} + V^*_i+ V^*_{i+1} 
\qquad \qquad 0 \leq i \leq \delta,
\label{eq:t2}
\end{equation}
where $V^*_{-1} = 0$ and $V^*_{\delta+1}= 0$.
\item[(iv)] There does not exist a subspace $W$ of $V$ such  that $AW\subseteq W$,
$A^*W\subseteq W$, $W\not=0$, $W\not=V$.
\end{itemize}
We say the pair $A,A^*$ is {\it over $\K$}.
We call $V$ the 
{\it underlying
 vector space}.
\end{defn}

\begin{note} According to a common 
notational convention, for a linear transformation $A$
the conjugate-transpose of $A$ is denoted
$A^*$. We emphasize we are {\it not} using this convention.
In a TD pair $(A,A^*)$, the linear transformations 
$A$ and $A^*$ are arbitrary subject to (i)--(iv) above.
\end{note}

Let $A,A^*$ denote a TD pair on $V$, as in Definition 
\ref{def:tdp}. By \cite[Lemma 4.5]{TD00} the integers $d$ and $\delta$ from
(ii), (iii) are equal and called the {\it diameter} of the
pair. An ordering of the eigenspaces of $A$ (resp. $A^*$)
is said to be {\em standard} whenever it satisfies  (\ref{eq:t1})
 (resp. (\ref{eq:t2})).  Let $\{V_i\}_{i=0}^d$ (resp.
$\{V^*_i\}_{i=0}^d$) denote a standard ordering of the eigenspaces
 of $A$ (resp. $A^*$). For $0 \leq i \leq d$, let 
$\theta_i$  (resp. $\theta^*_i$) denote the eigenvalue of
$A$  (resp.  $A^*$) associated with $V_i$  (resp. $V^*_i$). In what follows, we assume:
\begin{eqnarray}
\label{eq:const1}
&&\theta_i = \alpha + b q^{2i-d} + c q^{d-2i} 
\qquad \qquad 0 \leq i \leq d,
\\
\label{eq:const2}
&&\theta^*_i = \alpha^* + b^* q^{2i-d} + c^* q^{d-2i} 
\qquad \qquad 0 \leq i \leq d,
\label{eq:const3}
\end{eqnarray}
where $\alpha,\alpha^*$ are scalars, $q,\ b,\;b^*,\ c ,\;c^*$ are nonzero scalars such that $q\not=0, \ q^2 \not=1,
\ q^2 \not=-1$. In this case,  we say that $A,A^*$ is  a tridiagonal pair of $q-$Racah type \cite[Theorem~5.3]{NT:muqrac}.\vspace{1mm}

\begin{lem}\label{lem:polypart} For each positive integer $s$, there exists scalars $\beta_s,\gamma_s,{\gamma^*}_s,\delta_s,\delta^*_s$ in $\K$ such that
\beqa
\theta_i^2 - \beta_s\theta_i\theta_j  + \theta_j^2 - \gamma_s (\theta_i+\theta_j) -\delta_s&=&0 \ ,\\
{\theta_i^*}^2 - \beta_s{\theta_i^*}{\theta_j^*}  + {\theta_j^*}^2 - \gamma^*_s (\theta_i^*+\theta_j^*) -\delta^*_s &=&0 \quad \mbox{if} \quad |i-j|=s\nonumber \ . \quad (0\leq i,j\leq d).
\eeqa
\end{lem}

\vspace{2mm}

\subsection{Tridiagonal systems}
For the analysis presented below, it will be convenient to recall the notion of TD system \cite{TD00}.
Let ${\rm End}(V)$ denote the $\K$-algebra of all linear
transformations from $V$ to $V$. Let $A$ denote a diagonalizable element of $\mbox{\rm End}(V)$.
Let $\{V_i\}_{i=0}^d$ denote an ordering of the eigenspaces of $A$
and let $\{\theta_i\}_{i=0}^d$ denote the corresponding ordering of
the eigenvalues of $A$. For $0 \leq i \leq d$, define $E_i \in 
\mbox{\rm End}(V)$ 
such that $(E_i-I)V_i=0$ and $E_iV_j=0$ for $j \neq i$ $(0 \leq j \leq d)$.
Here $I$ denotes the identity of $\mbox{\rm End}(V)$. We call $E_i$ the {\em primitive idempotent} of $A$ corresponding to $V_i$
(or $\theta_i$). Observe that
(i) $I=\sum_{i=0}^d E_i$;
(ii) $E_iE_j=\delta_{i,j}E_i$ $(0 \leq i,j \leq d)$;
(iii) $V_i=E_iV$ $(0 \leq i \leq d)$;
(iv) $A=\sum_{i=0}^d \theta_i E_i$.
Here $\delta_{i,j}$ denotes the Kronecker delta.
Now let $A,A^*$ denote a TD pair on $V$.
An ordering of the primitive idempotents 
 of $A$ (resp. $A^*$)
is said to be {\em standard} whenever
the corresponding ordering of the eigenspaces of $A$ (resp. $A^*$)
is standard.
\begin{defn}
{\rm \cite[Definition~2.1]{TD00}}
 \label{def:TDsystem} 
\rm
Let $V$ denote a vector space over $\K$ with finite
positive dimension.
By a {\it tridiagonal system} (or {\it  $TD$ system}) on $V$ we mean a sequence
\[
 \Phi=(A;\{E_i\}_{i=0}^d;A^*;\{E^*_i\}_{i=0}^d)
\]
that satisfies (i)--(iii) below.
\begin{itemize}
\item[(i)]
$A,A^*$ is a TD pair on $V$.
\item[(ii)]
$\{E_i\}_{i=0}^d$ is a standard ordering
of the primitive idempotents of $A$.
\item[(iii)]
$\{E^*_i\}_{i=0}^d$ is a standard ordering
of the primitive idempotents of $A^*$.
\end{itemize}
We say that $\Phi$ is {\em over} $\K$.
\end{defn}

\begin{lem} 
{\rm \cite[Lemma~2.2]{TD10}}
\label{lem:triplep}
Let $(A; \lbrace E_i\rbrace_{i=0}^d; A^*; \lbrace E^*_i\rbrace_{i=0}^d)$
denote a TD system. Then the following hold for $0 \leq i,j,r\leq d$.
\begin{enumerate}
\item[\rm (i)] $E^*_iA^rE^*_j=0$ if $|i-j|>r$\ ,
\item[\rm (ii)] $E_iA^{*r}E_j=0$ if $|i-j|>r$\ .
\end{enumerate}
\end{lem}

\vspace{2mm}

\subsection{Higher order tridiagonal relations}
By \cite[Theorem 10.1]{TD00} the pair $A, A^*$ satisfy two polynomial
equations called the tridiagonal relations. These generalize the $q$-Serre relations (\ref{Uqserre})
and the Dolan-Grady relations \cite{DG}. 
\begin{thm}
\label{eq:lastchancedolangradyS99} (see \cite[Theorem 10.1]{TD00})
Let $\K$ denote a field, and let
 $A,A^*$ denote a TD pair of $q-$Racah type over  $\K$.
Then, there  exists a sequence of scalars $\beta, \gamma, \gamma^*,
\delta, \delta^*$ taken from $\K$ such that both
\begin{eqnarray}
\lbrack A,A^2A^*-\beta AA^*A + 
A^*A^2 -\gamma (AA^*+A^*A)-\delta A^*\rbrack &=& 0\ , 
\label{eq:qdolangrady199n}
\\
\lbrack A^*,A^{*2}A-\beta A^*AA^* + AA^{*2} -\gamma^* (A^*A+AA^*)-
\delta^* A\rbrack &=&0\ .  \qquad  \quad 
\label{eq:qdolangrady2S99n}
\end{eqnarray}
\end{thm}

\begin{proof}
Let $\Delta_1$ denote the expression of the left-hand side of (\ref{eq:qdolangrady199n}). We show $\Delta_1=0$. For $0\leq i,j \leq d$, one finds $E_i \Delta_1 E_j = p_1(\theta_i,\theta_j) \ E_i A^* E_j $ with
\beqa
p_1(x,y) = (x-y) (x^2 - \beta xy +y^2 - \gamma(x+y) -\delta)\ .
\eeqa
Take $\beta=\beta_1$, $\gamma=\gamma_1$ and $\delta=\delta_1$. For each $i,j$, according to Lemma \ref{lem:polypart} $p_1(\theta_i,\theta_j)=0$ if $|i-j|\leq 1$ and according to Lemma \ref{lem:triplep}, $E_i A^* E_j =0$  if $|i-j|> 1$. It implies (\ref{eq:qdolangrady199n}). Similar arguments are used to show (\ref{eq:qdolangrady2S99n}).
\end{proof}

Higher order relations can be constructed along the same line, using the properties of tridiagonal pairs described in the previous subsections. 
\begin{defn}
Let $x,y$ denote commuting indeterminates. For each positive integer $r$, we define the polynomials $p_r(x,y),p^*_r(x,y)$ as follows:
\beqa
p_r(x,y) = (x-y)\prod_{s=1}^{r} (x^2-\beta_s xy +y^2 - \gamma_s (x+y) -\delta_s)\ ,\label{defpoly}\\
p^*_r(x,y) = (x-y)\prod_{s=1}^{r} (x^2-\beta_s xy +y^2 - \gamma^*_s (x+y) -\delta^*_s)\ .
\eeqa
We observe $p_r(x,y)$ and $p^*_r(x,y)$ have a total degree $2r+1$ in $x,y$.
\end{defn}

\begin{lem}\label{lem:polyroot} For each positive integer $r$,
\beqa
p_r(\theta_i,\theta_j)=0 \quad  \mbox{and} \quad  p^*_r(\theta^*_i,\theta^*_j)=0 \qquad \mbox{if} \quad |i-j|\leq r \ , (0\leq i,j\leq d)\ .
\eeqa
\end{lem}
\begin{proof}
Immediate.
\end{proof}
\begin{thm}
\label{hqdg} 
For each positive integer $r$,
\beqa
\sum_{i,j=0}^{i+j\leq 2r+1} a_{ij} A^i {A^*}^rA^j =0 \ , \qquad \sum_{i,j=0}^{i+j\leq 2r+1} a^*_{ij} {A^*}^i {A}^r{A^*}^j =0 \  \label{hqdgr}
\eeqa
where the scalars $a_{ij} ,a^*_{ij}$ are defined by:
\beqa
p_r(x,y)=  \sum_{i,j=0}^{i+j\leq 2r+1} a_{ij} x^i y^j \quad  \mbox{and} \quad  p^*_r(x,y)=\sum_{i,j=0}^{i+j\leq 2r+1} a^*_{ij} x^i y^j \ .\label{polyr}
\eeqa
\end{thm}
\begin{proof}
Let $\Delta_r$ denote the expression of the left-hand side of the first equation of (\ref{hqdgr}). We show $\Delta_r=0$. For $0\leq i,j \leq d$, one finds $E_i \Delta_r E_j = p_r(\theta_i,\theta_j) \ E_i {A^*}^r E_j $ with (\ref{polyr}). According to Lemma \ref{lem:polyroot} and Lemma  \ref{lem:triplep}, it follows $\Delta_r=0$. Similar arguments are used to show the second equation of (\ref{hqdgr}).
\end{proof}

\subsection{Higher order $q-$Dolan-Grady relations} As a straightorward application, we are now interested by the $r-th$ higher order tridiagonal relations which special cases for $r=1$ produce the defining relations of the $q-$Onsager algebra (\ref{qDG}). 

\begin{defn} Consider a TD pair $A,A^*$  of $q-$Racah type with eigenvalues such that $\alpha=\alpha^*=0$.  Assume $r=1$.  The corresponding tridiagonal relations (\ref{eq:qdolangrady199n}), (\ref{eq:qdolangrady2S99n}) are called the $q-$Dolan-Grady relations.
\end{defn}

\begin{example} For a TD pair $A,A^*$  of $q-$Racah type with eigenvalues such that $\alpha=\alpha^*=0$,  
 the parameter sequence is given by $\beta_1= q^{2} + q^{-2}\ , \ \gamma_1=\gamma^*_1=0\ , \   \delta_1= - bc(q^{2} - q^{-2})^2 \ ,\ \delta^*_1= - b^*c^*(q^{2} - q^{-2})^2$. Define $\delta_1=\rho_0$, $\delta_1^*=\rho_1$. The $q-$Dolan-Grady relations
are given by:
\ben
\quad \sum_{j=0}^{3} (-1)^j  \left[ \begin{array}{c} 3 \\  j \end{array}\right]_q   A^{3-j} {A^*} A^{j} - \rho_0 (AA^*-A^*A) &=& 0\ ,\label{qDG1p} \\
\sum_{j=0}^{3} (-1)^j  \left[ \begin{array}{c} 3 \\ j \end{array}\right]_q   {A^*}^{3-j} {A} {A^*}^{j} - \rho_1 (A^*A-AA^*) &=& 0\ . \label{qDG2p}
\een
\end{example}
\begin{rem} The relations (\ref{qDG1p}), (\ref{qDG2p}) are the defining relations of the $q-$Onsager algebra (\ref{qDG}).
\end{rem}

\begin{defn}  Consider a TD pair $A,A^*$  of $q-$Racah type with eigenvalues such that $\alpha=\alpha^*=0$.  For any positive integer $r$, the corresponding higher order tridiagonal relations (\ref{hqdgr}) are called the higher order $q-$Dolan-Grady relations.
\end{defn}

\begin{example}\label{exfinr} For a TD pair $A,A^*$  of $q-$Racah type with eigenvalues such that $\alpha=\alpha^*=0$,  the higher order $q-$Dolan-Grady relations are given by (\ref{qDGfinr}) with the identification $\rho_0=\delta_1$ and $\rho_1=\delta^*_1$.
\end{example}
\begin{proof}  For $\alpha=\alpha^*=0$ in (\ref{eq:const1}), (\ref{eq:const2}), from Lemma {\ref{lem:polypart}} one finds $\beta_s= q^{2s} + q^{-2s}$\ , \ $\gamma_s=\gamma^*_s=0$\ , \   $\delta_s= - bc(q^{2s} - q^{-2s})^2$ \ ,\ $\delta^*_s= - b^*c^*(q^{2s} - q^{-2s})^2$ \ . Then, the first polynomial generating function (\ref{defpoly}) reads:
\beqa
p_r(x,y) = (x-y)\prod_{s=1}^{r} \left(x^2- \frac{[2s]_{q^2}}{[s]_{q^2}}xy +y^2  -   [s]^2_{q^2} \rho_0\right)\ ,\label{polyqons}
\eeqa
where the notation $\beta_s=[2s]_{q^2}/[s]_{q^2}$ and $\delta_s/\rho_0=[s]^2_{q^2}$ has been introduced. Expanding the polynomial in the variables $x,y$ as (\ref{polyr}), one shows that the coefficients $a_{ij}$ in (\ref{polyr}) take the form:
\beqa
 a_{2r+1-2p - j \ j}= (-1)^{j+p} \rho^{p}_0 c_j^{[r,p]}
\eeqa
where $c_j^{[r,p]}$  solely depend on $q$, and are vanishing otherwise. By induction, one finds that they are given by (\ref{cfinr}).
Replacing $\rho_0\rightarrow \rho_1$, the second relation in (\ref{qDGfinr})  follows.
\end{proof}

For $r=2,3$, the higher order $q-$Dolan-Grady relations (\ref{qDGfinr}) can be constructed in a straightforward manner: 
\begin{example}\label{exr2} The first example of higher order $q-$Dolan-Grady relations is given by (\ref{qDGfinr})  for $r=2$ with:
\beqa
c^{[2,0]}_0&=&1\ ,\quad c^{[2,0]}_1= 1+ [2]_{q^2} + \frac{[4]_{q^2}}{[2]_{q^2}} \equiv \left[ \begin{array}{c} 5\\ 1 \end{array}\right]_q  \ ,\qquad c^{[2,0]}_2= 2+ [2]_{q^2}+ [4]_{q^2} + \frac{[4]_{q^2}}{[2]_{q^2}} \equiv  \left[ \begin{array}{c} 5\\ 2 \end{array}\right]_q \ ,\nonumber \\
c^{[2,1]}_0&=& 1 + [2]^2_{q^2} \equiv q^4+q^{-4} +3 \ , \qquad  c^{[2,1]}_1=  1 + [2]^2_{q^2} +  \frac{[4]_{q^2}}{[2]_{q^2}} +  [2]^3_{q^2} \equiv [5]_q[3]_q\ ,\nonumber\\
c^{[2,2]}_0&=& [2]^2_{q^2} \equiv (q^2+q^{-2})^2 \ . \nonumber 
\eeqa
\end{example}
\begin{example}\label{exr3} The second example of higher order $q-$Dolan-Grady relations is given by (\ref{qDGfinr})  for $r=3$ with:
\beqa
c_{j}^{[3,0]} &=&  \left[ \begin{array}{c} 7 \\ j 
\end{array}\right]_q , \quad j=0,...,7\ ,\nonumber\\     
c^{[3,1]}_0&=& 1 + [2]^2_{q^2} + [3]^2_{q^2} \ ,\nonumber\\
c^{[3,1]}_1&=& 1 + [2]^2_{q^2} + [3]^2_{q^2} + (1+ [3]^2_{q^2})\frac{[4]_{q^2}}{[2]_{q^2}} +  ([2]^2_{q^2} + [3]^2_{q^2})[2]_{q^2}  + (1+ [2]^2_{q^2})\frac{[6]_{q^2}}{[3]_{q^2}} \ ,\nonumber\\
c^{[3,1]}_2&=& 2(1 + [2]^2_{q^2} + [3]^2_{q^2}) + (1+ [3]^2_{q^2})\frac{[4]_{q^2}}{[2]_{q^2}} +  ([2]^2_{q^2} + [3]^2_{q^2})[2]_{q^2}  + (1+ [2]^2_{q^2})\frac{[6]_{q^2}}{[3]_{q^2}} \ \nonumber\\
&&+\ \  \frac{[4]_{q^2}[6]_{q^2}}{[2]_{q^2}[3]_{q^2}} + \frac{[2]^3_{q^2}[6]_{q^2}}{[3]_{q^2}}  + [3]^2_{q^2}[4]_{q^2} \ ,\nonumber\\
c^{[3,2]}_0&=& [2]^2_{q^2} + [3]^2_{q^2} +  [2]^2_{q^2}[3]^2_{q^2}\ ,\nonumber\\
c^{[3,2]}_1&=& [2]^2_{q^2} + [3]^2_{q^2} +  [2]^2_{q^2}[3]^2_{q^2} + [2]^2_{q^2}\frac{[6]_{q^2}}{[3]_{q^2}} + [3]^2_{q^2}\frac{[4]_{q^2}}{[2]_{q^2}} + [2]^3_{q^2}[3]^2_{q^2} ,\nonumber\\
c^{[3,3]}_0&=& [2]^2_{q^2}[3]^2_{q^2} \ .\nonumber
\eeqa
\end{example}
%
%
\vspace{2mm}

To end up this Section, let us consider the family of relations satisfied by a TD pair of $q-$Racah type such that\footnote{For instance, choose $b,b^*=0$ and/or $c,c^*=0$ in (\ref{eq:const1}), (\ref{eq:const2}).} $\rho_0=\rho_1=0$. In this special case, the  polynomial generating function for the coefficients (\ref{polyqons}) factorizes:
\beqa
p_r(x,y) = x^{2r+1}\prod_{s=-r}^{r} \left(1- q^{2s}\frac{y}{x}\right)\ .\label{polyqserre}
\eeqa
According to the $q-$binomial theorem, one has:
\beqa
(1-u)(1-q^2u) \cdot \cdot \cdot (1-q^{4r}u) = \sum_{j=0}^{2r+1}   \left[ \begin{array}{c} 2r+1\\ j \end{array}\right]_q \ (-1)^j q^{2jr} u^j \ . \label{qbinom}
\eeqa
If we denote the l.h.s. of (\ref{qbinom}) by $g_{q}(u)$, observe $p_r(x,y) = x^{2r+1}g_{q}(q^{-2r}\frac{y}{x})$. As a consequence, for $\rho_0=\rho_1=0$ the higher order $q-$Dolan-Grady relations satisfied by the corresponding TD pair simplify to the well-known Lusztig's higher order  $q-$Serre relations \cite{Luzt}:
\beqa
\sum_{j=0}^{2r+1}  (-1)^{j}  \left[ \begin{array}{c} 2r+1\\ j \end{array}\right]_q  \,  A^{2 r+1-j} {A^*}^{r} A^{j}&=&0\  \label{qserrefinr} \ ,\\
\sum_{j=0}^{2r+1} (-1)^{j}  \left[ \begin{array}{c} 2r+1\\ j \end{array}\right]_q  \,  {A^*}^{2 r+1-j} {A}^{r} {A^*}^{j}&=&0\  .
\ \nonumber
\eeqa
\vspace{2mm}

\section{Higher order relations for the $q-$Onsager algebra; recursion for generating the coefficients}
In the previous Section, it was shown that every TD pair of $q-$Racah type such that $\alpha=\alpha^*=0$ satisfies the $r-th$ higher order $q-$Dolan-Grady relations (\ref{qDGfinr}) with (\ref{cfinr}). For the special case $r=1$, these relations coincide with the defining relations of the $q-$Onsager algebra (\ref{qDG}). According to these results, we propose the following conjecture:
\begin{conj}\label{conj1} Let $A$, $A^*$ be the fundamental generators of the $q-$Onsager algebra (\ref{qDG}).  The relations (\ref{qDGfinr}) with (\ref{cfinr}) hold.
\end{conj}
The purpose of this Section is to prove the conjecture for $r=2,3$. Then, using an inductive argument we will study the general structure and derive recursion formulae - independently of the results of the previous Section - for the coefficients $c^{[r,p]}_j$. As we will discuss at the end of this Section, a detailed comparison with the coefficients obtained in the previous Section for several values of $r\geq 4$ supports the conjecture. \vspace{1mm}

Let $A$, $A^*$ be the fundamental generators of the $q-$Onsager algebra (\ref{qDG}). Observe that the defining relations can be written:
\ben
\quad \sum_{i=0}^{3} (-1)^i  \left[ \begin{array}{c} 3 \\  i \end{array}\right]_q   A^{3-i} {A^*} A^{i} - \rho_0 (AA^*-A^*A) &=& 0\ ,\label{qDG1} \\
\sum_{i=0}^{3} (-1)^i  \left[ \begin{array}{c} 3 \\ i \end{array}\right]_q   {A^*}^{3-i} {A} {A^*}^{i} - \rho_1 (A^*A-AA^*) &=& 0\ . \label{qDG2}
\een
By analogy with Lusztig's higher order $q-$Serre relations, we are interested in more complicated linear combinations of monomials  of the type $A^n{A^*}^rA^m$, $n+m=2r+1,2r-1,...,1$,  that are vanishing. The defining relations (\ref{qDG1}), (\ref{qDG2}) correspond to the case $r=1$ of (\ref{qDGfinr}). Below, successively we derive the relations (\ref{qDGfinr}) for $r=2,3$ and study the generic case by induction. 

\subsection{Proof of the relations for $r=2$} Consider the simplest example beyond (\ref{qDG1}): we are looking for a linear relation between monomials of the type $A^n{A^*}^2A^m$, $n+m=5,3,1$. According to the defining relations (\ref{qDG1}), note that the monomial $A^3A^*$ can be written as:
\beqa
A^3A^* = \al A^2A^*A - \al AA^*A^2 + A^*A^3 + \rho_0 (AA^* - A^*A)\ \quad \mbox{with}\quad \alpha=[3]_q \ .\label{mon1}
\eeqa
Multiplying from the left by $A$ or $A^2$, the corresponding expressions can be ordered as follows: each time a monomial of the form $A^n A^*A^m$ with $n\geq 3$ arise, it is reduced using (\ref{mon1}). It follows:
\beqa
A^4A^* &=& (\al^2-\al) A^2A^*A^2 + (1- \al^2) AA^*A^3 + \al A^*A^4 + \rho_0 (A^2A^* - \al A^*A^2 + (\al-1)AA^*A)\ ,\nonumber\\
A^5A^* &=& (\al^3-2\al^2+1) A^2A^*A^3 + \al(-\al^2+\al+1) AA^*A^4 + \al(\al-1) A^*A^5 \nonumber\\
&&+  \ \rho_0 \left((2\al-1)A^2A^*A + \al(\al-3)AA^*A^2 - (\al^2-\al-1) A^*A^3\right) \nonumber \\
&&+ \ \rho_0^2(AA^*-A^*A)\ .\nonumber
\label{mon2}
\eeqa
For our purpose, four different types of monomials may be now considered: $A^5{A^*}^2$, $A^4{A^*}^2A$, $A^3{A^*}^2A^2$ and $A^3{A^*}^2$. Following the ordering prescription, each of these monomials can be reduced as a combination of monomials of the type ($n,m,p,s,t\geq 0$):
\beqa
&&\quad A^n{A^*}^2A^m \ \qquad \ \mbox{with} \quad  \ n\leq 2\ , \ n+m=5,3,1\ ,\label{mongen}\\
 &&\quad A^p A^* A^s A^* A^t \quad \  \mbox{with} \quad  \ p\leq 2\ , \ s\geq 1\ ,\ p+s+t=5,3,1 \ .\nonumber
\eeqa
For instance, the monomial  $A^5{A^*}^2$ is reduced to:
\beqa
A^5{A^*}^2 &=& (\al^3-2\al^2+1) \left(\al A^2A^*A^2A^*A - \al A^2A^*AA^*A^2 + A^2{A^*}^2 A^3\right) \nonumber\\
&&-\ (\al^3-\al^2-\al) \left((\al^2-\al)AA^*A^2A^*A^2 +(1-\al^2) AA^*AA^*A^3  + \al A{A^*}^2A^4 \right)\nonumber\\
&&+ \ (\al^2-\al)\left( (\al^3-2\al^2+1)  {A^*}A^2A^*A^3  - (\al^3-\al^2-1)  {A^*}AA^*A^4 + \al(\al-1) {A^*}^2A^5  \right)\nonumber \\
&& + \ \rho_0 (\al^3-2\al^2+2\al)\left(A^2A^*AA^* - AA^*A^2A^* + A^*A^2A^*A\right)\nonumber\\
&& + \ \rho_0 (-\al^3+2\al^2-1)\left(A^2{A^*}^2A + \al AA^*AA^*A\right)\nonumber\\
&& + \ \rho_0 (\al^4-\al^3-\al^2)\left(A{A^*}^2A^2\right)\nonumber\\
&& + \ \rho_0 (\al^4-3\al^3+2\al^2-\al)\left(A^*{A}A^*A^2\right)\nonumber\\
&& + \ \rho_0 (-\al^4+2\al^3-\al^2+1)\left({A^*}^2A^3\right) \nonumber\\
&& + \ \rho_0^2 \big[A,{A^*}^2\big] \ . \nonumber
\eeqa
The two other monomials $A^4{A^*}^2A$, $A^3{A^*}^2A^2$ are also ordered using (\ref{mon1}). One obtains:
\beqa
A^4{A^*}^2A &=& (\al^2-\al) \left( A^2A^*A^2A^*A + \al A^*A^2A^*A^3\right) + \ \al^2{A^*}^2A^5   \nonumber\\
&&+\ (\al^2-1) \left(\al AA^*AA^*A^3- \al AA^*A^2A^*A^2 - \al A^*AA^*A^4 - A{A^*}^2A^4\right)\nonumber\\
&& + \ \rho_0 \left(  A^2{A^*}^2 A  - (1-\al^2) A{A^*}^2A^2 - \al^2 {A^*}^2 A^3\right)
\nonumber \\
&&+ \ \rho_0 \left(\al^2-\al \right) \left(A^*AA^*A^2-AA^*AA^*A \right)\ ,\nonumber \\
A^3{A^*}^2A^2 &=& \al \left(A^2A^*AA^*A^2 - AA^*A^2A^*A^2 + A^*A^2A^*A^3 - A^*AA^*A^4  \right) +{A^*}^2A^5 + \rho_0(A {A^*}^2A^2 - {A^*}^2A^3) \nonumber \ .
\eeqa
The ordered expression for the fourth monomial $A^3{A^*}^2$ directly follows from (\ref{mon1}). 
Having the explicit ordered expressions of $A^5{A^*}^2$, $A^4{A^*}^2A$, $A^3{A^*}^2A^2$ and $A^3{A^*}^2$ in terms of monomials of the type (\ref{mongen}), let us consider the combination:
\beqa
f_2(A,A^*)= c_0^{[2,0]} A^5{A^*}^2 - c^{[2,0]}_1 A^4{A^*}^2A +   c^{[2,0]}_2 A^3{A^*}^2A^2 -  \rho_0 c^{[2,1]}_0 A^3{A^*}^2 \ 
\eeqa
with unknown coefficients $c^{[2,0]}_j$, $j=1,2$, $c^{[2,1]}_0$, and normalization $c_0^{[2,0]}=1$.
After simplifications, the combination takes the ordered form:
\beqa
f_2(A,A^*)=   c^{[2,0]}_3 A^2{A^*}^2A^3 - c^{[2,0]}_4 A{A^*}^2A^4 + c^{[2,0]}_5 {A^*}^2A^5  \ + \ g_2(A,A^*)  \ 
\eeqa
where 
\beqa
 c^{[2,0]}_3 = \al^3-2\al^2+1 \ , \quad c^{[2,0]}_4 = \al^2(\al^2-\al-1)+ c^{[2,0]}_1 (1-\al^2) \ , \quad c^{[2,0]}_5 = (\al^2-\al)^2 - \al^2c^{[2,0]}_1 + c^{[2,0]}_2 \ .\nonumber
\eeqa
Inspired by the structure of Lusztig's higher order $q-$Serre relations, consider the conditions under which the combination $g_2(A,A^*)$ never contains monomials of the form $A^p A^* A^s A^* A^t $ ($p\leq 2,\ s\geq 1$). At the lowest order in $\rho_0$, given a particular monomial  the condition under which its coefficient in $g_2(A,A^*)$ is vanishing is given by:
\beqa
&& A^2A^*A^2A^*A: \ \al^3-2\al^2+1 - c^{[2,0]}_1(\al-1)=0 \ ,\nonumber\\
&& A^2A^*AA^*A^2: \ -\al^3+2\al^2-1 + c^{[2,0]}_2 =0\ , \nonumber \\
&& AA^*A^2A^*A^2:  \ (\al-1)(-\al^3+\al^2+\al)  - c^{[2,0]}_1(1-\al^2) - c^{[2,0]}_2 =0 \ ,\nonumber\\
&& AA^*AA^*A^3:  \ (1-\al^2)(-\al^3+\al^2+\al)+ c^{[2,0]}_1 \al(1-\al^2) =0\ , \nonumber\\
&& A^*A^2A^*A^3:  \  (\al-1)(\al^3-2\al^2+1)  - c^{[2,0]}_1(\al^2-\alpha) + c^{[2,0]}_2 =0 \ ,\nonumber\\
&& A^*AA^*A^4:  \ (\al-1)(-\al^3+\al^2+\al)- c^{[2,0]}_1 (1-\al^2) - c^{[2,0]}_2 =0\ , \nonumber
\eeqa
Recall that $\al=[3]_q$. The solution  $\{c^{[2,0]}_j$, $j=1,2$\} to this system of equations exists, and it is unique. In terms of $q-$binomials, it reads:
\beqa
c^{[2,0]}_1= \left[ \begin{array}{c} 5\\ 1 \end{array}\right]_q  \ ,\qquad c^{[2,0]}_2= \left[ \begin{array}{c} 5\\ 2 \end{array}\right]_q \ .
\eeqa
At the next order $\rho_0$, the conditions such that monomials of the type $A^p A^* A^s A^* A^t $ with $p+s+t\leq 3$ and $s\geq 1$ yield to:
\beqa
c^{[2,1]}_0=q^4+q^{-4}+3\ .\nonumber
\eeqa
All other coefficients of the monomials $A^n{A^*}^2A^m \quad \mbox{for} \quad  n+m=5,3,1$ are explicitly determined in terms of $c^{[2,0]}_j$ ($j=0,1,2$), $c^{[2,1]}_0$, $\rho_0$ and $\rho_0^2$. Based on these results, we conclude that the $q-$Dolan-Grady relation (\ref{qDG1}) implies the existence of a unique linear relation between monomials of the type $A^n{A^*}^2A^m \quad \mbox{with} \quad  n+m=5,3,1$. This relation can be seen as a $\rho_0-$deformed analogue of the simplest higher order $q-$Serre relation.  Explicitly, one finds:
\beqa
\qquad \sum\limits_{j = 0}^5 {{{\left( { - 1} \right)}^j}\left[ {\begin{array}{*{20}{c}}
   5  \\
   j  \\
\end{array}} \right]} {A^{5 - j}}{A^*}^2{A^j} &=& 
\rho_0 \left( (q^4 + q^{- 4} + 3)(A^3 {A^*}^2-  {A^{*2}}{A^3}) - \left[ 5 \right]_q\left[ 3 \right]_q(A^2{A^{*2}}A - A{A^{*2}}A^2) \right) \label{qDG11} \\ 
& &- \ {\rho_0 ^2}{\left( {{q^2} + {q^{ - 2}}} \right)^2}\left(  A{A^{*2} - A^{*2}}A  \right)\ .\nonumber
\eeqa
Using the automorphism $A\leftrightarrow A^*$ and $\rho_0\leftrightarrow \rho_1$  which exchanges (\ref{qDG1}) and (\ref{qDG2}), the second relation generalizing (\ref{qDG2}) is obtained.  The coefficients coincide with the ones given in Example \ref{exr2}, which proves conjecture \ref{conj1} for $r=2$.
For the special undeformed case $\rho_0=\rho_1=0$, note that both relations reduce to the simplest examples of higher order $q-$Serre relations. \vspace{1mm}

\subsection{Proof of the relations for $r=3$} Following a similar analysis, the next  example of higher order $q-$Dolan-Grady relations can be also derived. To this end, one is looking for a linear relation between monomials of the type $A^n{A^*}^3A^m$, $n+m=7,5,3,1$. Assume the $q-$Dolan-Grady relation (\ref{qDG1}) and its simplest consequence (\ref{qDG11}). Write the four monomials:
\beqa
A^7 {A^*}^3 = (A^7{A^*}^2)A^*\ , \quad  A^6{A^*}^3 A = (A^6{A^*}^2)A^*A \ ,\quad  A^5{A^*}^3A^2 = (A^5 {A^*}^2)A^*A^2\ ,\quad A^5{A^*}^3 = (A^5 {A^*}^2)A^*     \ .\nonumber
\eeqa
Using (\ref{qDG11}) and then (\ref{qDG1}), they can be expressed solely in terms of   monomials of the type:
\beqa
&&\qquad A^n{A^*}^3A^m \quad \qquad \mbox{with} \quad n\leq 4\ , \ n+m=7,5,3,1\ , \label{mongen3}\\
&&\qquad A^p {A^*}^2 A^s A^* A^t \quad  \mbox{with} \quad  p\leq 4\ ,\ s\geq 1\ ,\ p+s+t=7,5,3,1 \ .\nonumber
\eeqa
Then, introduce the combination 
\beqa
f_3(A,A^*)= c_0^{[3,0]} A^7{A^*}^3 - c^{[3,0]}_1 A^6{A^*}^3A +  c^{[3,0]}_2 A^5{A^*}^3A^2 -  \rho_0 c^{[3,1]}_0 A^5{A^*}^3\   \label{comb3}
\eeqa
with unknown coefficients $c^{[3,0]}_j$, $(j=1,2)$, $c^{[3,1]}_0$  and normalization $c_0^{[3,0]}=1$. By straightforward calculations using the ordered expressions of $A^7{A^*}^3$, $A^6{A^*}^3A$, $A^5{A^*}^3A^2$ and $A^5{A^*}^3$, $f_3(A,A^*)$ is reduced to a combination of monomials of the type (\ref{mongen3}). Note that the coefficients of the monomials $A^n{A^*}^3A^m$ for $n+m=5,3,1$ are of order $\rho_0,\rho_0^2, \rho_0^3$, respectively. Identifying the conditions under which the coefficient of any monomial of the form
\beqa
A^p {A^*}^2 A^s A^* A^t \quad \mbox{with} \quad p\leq 4\ ,\ s\geq 1\ , \quad  \ p+s+t=7,5,3,1\ , \nonumber 
\eeqa
is vanishing, one obtains a system of equations for the coefficients, which solution is unique. Simplifying (\ref{comb3}) according to the explicit solutions $c^{[3,0]}_j$, $j=1,2$ and $c^{[3,1]}_0$, one ends up with the next example of higher order $q-$Dolan-Grady relations. Using the automorphism $A\leftrightarrow A^*$ and $\rho_0\leftrightarrow \rho_1$,
the second relation follows. One finds:  
\ben
&&\sum_{p=0}^{3}\, \rho_0^{p} \,\sum_{j=0}^{7-2p} (-1)^{j+p}  \,c_{j}^{[3,p]}\,  A^{7-2p-j} {A^*}^3 A^{j}=0 \ , \ \label{qDG13}\\
&&\sum_{p=0}^{3}\, \rho_1^{p}\, \sum_{j=0}^{7-2p} (-1)^{j+p} \, c_{j}^{[3,p]} \, {A^*}^{7-2p-j} { A}^3 {A^*}^{j}=0\ \nonumber
\een
where $c_{j}^{[3,p]}=c_{7-2p-j}^{[3,p]}$\ , \ $c_{j}^{[3,0]} =  \left[ \begin{array}{c} 7 \\ j 
\end{array}\right]_q$ and     
\beqa
&& c^{[3,1]}_0= \left( {{q^8} + 3{q^4} + 6 + 3{q^{ - 4}} + {q^{ - 8}}} \right)\ ,\qquad  c^{[3,1]}_1={\left[ 7 \right]_q}\left( {{q^6} + {q^4} + {q^2} + 4 + {q^{ - 2}} + {q^{ - 4}} + {q^{ - 6}}} \right)\ ,\nonumber\\
&&  c^{[3,1]}_2=  \left[ 7 \right]_q (q^2-1+q^{-2})(q^2+q^{-2})(q^4+2q^2+4 + 2q^{-2}+q^{-4})\ ,\nonumber\\
&& c^{[3,2]}_0= ( q^{6} + 2q^2 + 3q^{-2}+ q^{-6}) ( q^{6} + 3q^2 + 2q^{-2}+ q^{-6}) \ ,\nonumber\\
&&  c^{[3,2]}_1=\left[ 7 \right]_q \left( {q^{8}} + {q^{6}} + 4{q^{4}} + {q^2} + 7 + {q^{-2}} + 4{q^{-4}} + q^{-6} + {q^{ - 8}} \right)\ ,\nonumber\\
&& c^{[3,3]}_0=  [2]^2_{q^2}[3]^2_{q^2}\ .\nonumber
\eeqa
It is straightforward to compare the coefficients above with the ones obtained from the expansion of the polynomial generating function $p_3(x,y)$. Although the coefficients above look different, they coincide exactly with the ones reported in Example \ref{exr3}. This proves the conjecture for $r=3$. Again,  for the special case $\rho_0=\rho_1=0$ Lusztig's higher order $q-$Serre relations are recovered. 

\subsection{Relations for $r$ generic} Above examples (\ref{qDG11}), (\ref{qDG13}) suggest that conjecture {\ref{conj1}} holds for $r$ generic.  Looking for a linear relation between monomials of the type $A^n{A^*}^rA^m$, $n+m=2r+1,2r-1,...,1$, 
for $r\geq 1$, relations of the form 
\ben
&&\sum_{p=0}^{r}\, \rho_0^{p} \,\sum_{j=0}^{2(r-p)+1} (-1)^{j+p}  \,c_{j}^{[r,p]}\,  A^{2(r-p)+1-j} {A^*}^r A^{j}=0 \ , \ \label{qDGr}\\
&&\sum_{p=0}^{r}\, \rho_1^{p}\, \sum_{j=0}^{2(r-p)+1} (-1)^{j+p} \, c_{j}^{[r,p]} \, {A^*}^{2(r-p)+1-j} { A}^r {A^*}^{j}=0\ \label{qDGr2}
\een
are expected, provided the elements $A,A^*$ satisfy the $q-$Dolan-Grady relations (\ref{qDG1}) and (\ref{qDG2}). Our aim is now to study these relations in details and obtain recursive formulae for the coefficients $c_{j}^{[r,p]}$. \vspace{1mm} 

In order to study the higher order $q-$Dolan-Grady relations (\ref{qDGr}) for generic values of $r$, we proceed by induction. First, assume  the basic relation (\ref{qDG1}) holds and implies all relations (\ref{qDGr}) up to $r$ which explicit coefficients $c_{j}^{[r,p]}$ in terms of $q$ are assumed to be known. It is the case for $r=2,3$ as shown above. Our aim is to construct the higher order relation associated with $r+1$ and express the coefficients $c_{j}^{[r+1,p]}$, ($j=0,1,...2r+3-2p$, $p=0,...,r+1$) in terms of $c_{j'}^{[r,p']}$ ($j'=0,1,...2r+1-2p'$, $p'=0,...,r$) . Following the steps described for $r=2,3$, from the relation (\ref{qDGr}) we first deduce:
\beqa
A^{2r+1}{A^*}^r &=&-\sum\limits_{j=1}^{2r+1}{(-1)^jc^{[r,0]}_jA^{2r+1-j}{A^*}^rA^j}-\sum\limits_{p=1}^r{\rho_0^p\sum\limits_{j=0}^{2(r-p)+1}{(-1)^{j+p}c^{[r,p]}_jA^{2(r-p)+1-j}{A^*}^rA^j}} \label{eq1} \ ,\\
A^{2r+2}{A^*}^r &=&-\sum\limits_{j=2}^{2r+2}{(-1)^jM^{(r,0)}_jA^{2r+2-j}{A^*}^rA^j}-\sum\limits_{p=1}^r{\rho_0^p\sum\limits_{j=0}^{2(r-p)+2}{(-1)^{j+p}M^{(r,p)}_jA^{2(r-p)+2-j}{A^*}^rA^j}}\label{eq2}\ ,\\ 
A^{2r+3}{A^*}^r &=&-\sum\limits_{j=3}^{2r+3}{(-1)^jN^{(r,0)}_jA^{2r+3-j}{A^*}^rA^j}-\sum\limits_{p=1}^r{\rho_0^p\sum\limits_{j=0}^{2(r-p)+3}{(-1)^{j+p}N^{(r,p)}_jA^{2(r-p)+3-j}{A^*}^rA^j}} \label{eq3}
\eeqa
where the coefficients $M^{(r,p)}_j$, $N^{(r,p)}_j$ are determined recursively in terms of $c_{j}^{[r,p]}$ (see Appendix A). Now, write the four monomials:
\beqa
&& A^{2r+3} {A^*}^{r+1} = (A^{2r+3}{A^*}^r)A^*\ , \quad A^{2r+2} {A^*}^{r+1}A = (A^{2r+2}{A^*}^r)A^*A \ ,\quad  A^{2r+1} {A^*}^{r+1}A^2 = (A^{2r+1}{A^*}^{r})A^*A^2  \  ,\nonumber \\
\qquad && A^{2r+1}{A^*}^{r+1} = (A^{2r+1}{A^*}^r)A^* \ . \nonumber
\eeqa
Using (\ref{eq1})-(\ref{eq3}), they can be expressed solely in terms of:
\beqa
\qquad \quad && A^n{A^*}^{r+1}A^m \qquad\quad \ \mbox{with} \quad  n\leq 2r \ ,\ n+m=2r+3, 2r+1,...,1\ ,\label{mongenr}\\
  \quad && A^p {A^*}^r A^s A^* A^t \quad \quad  \ \mbox{with} \quad  p\leq 2r \ ,\ s\geq 1\ , \ p+s+t=2r+3, 2r+1,...,1 \ .\label{mongenUNr}
\eeqa
It is however clear from (\ref{eq1})-(\ref{eq3}) that each monomial  $A^{2r+3} {A^*}^{r+1}$, $A^{2r+2} {A^*}^{r+1}A$, $A^{2r+1} {A^*}^{r+1}A^2$ and $A^{2r+1}{A^*}^{r+1}$   can be further reduced using (\ref{qDG1}). For instance, 
\beqa
 (A^{2r+3}{A^*}^r)  A^*=-\sum\limits_{j=3}^{2r+3}(-1)^jN^{(r,0)}_jA^{2r+3-j}{A^*}^r\!\!\!\!\!\!\!\underbrace{A^j A^*}_{\mbox{reducible}}-\ \ \sum\limits_{p=1}^r \rho_0^p\sum\limits_{j=0}^{2(r-p)+3}(-1)^{j+p}N^{(r,p)}_jA^{2(r-p)+3-j}{A^*}^r \!\!\!\!\!\!\!\!\!\!\!\!\!\!\!\!\!\underbrace{A^j A^*}_{\mbox{reducible if $j\geq 3$}} \!\!\!\!\!\!\!\!\!\!\!\!\!\!\!\!\ .\label{reducible}\nonumber
\eeqa
According to (\ref{qDG1}), observe that the monomials $A^jA^*$ (for $j$ even or odd) can be written as:
\beqa
A^{2n+2}A^*&=&\sum\limits_{k=0}^n{\sum\limits_{i=0}^2{\rho_0^{n-k}\eta^{(2n+2)}_{k,i}A^{2-i}A^*A^{2k+i}}} \ ,\label{eqmu1} \\
A^{2n+3}A^*&=&\sum\limits_{k=1}^{n+1}{\sum\limits_{i=0}^2{\rho_0^{n+1-k}\eta^{(2n+3)}_{k,i}A^{2-i}A^*A^{2k-1+i}}}+\rho_0^{n+1}(AA^*-A^*A) \ ,\label{eqmu2}
\eeqa
where the coefficients $\eta^{(2n+2)}_{k,i},\eta^{(2n+3)}_{k,i}$ are determined recursively in terms of $q$ (see Appendix A). It follows:
\beqa
&& \! \! \! \! \! \! \! \!\!\!\!\! \qquad A^{2r+3}{A^*}^{r+1}=\sum\limits_{i=1}^{r+1}{N^{(r,0)}_{2i+1}A^{2(r-i)+2}{A^*}^r(\sum\limits_{k=1}^i{\sum\limits_{j=0}^2{\rho_0^{i-k}\eta^{(2i+1)}_{k,j}}A^{2-j}A^*A^{2k-1+j}}+\rho_0^i(AA^*-A^*A))} \label{m1}\\ & &\qquad - \sum\limits_{i=1}^r{N^{(r,0)}_{2i+2}A^{2(r-i)+1}{A^*}^r(\sum\limits_{k=0}^i{\sum\limits_{j=0}^2{\rho_0^{i-k}\eta^{(2i+2)}_{k,j}A^{2-j}A^*A^{2k+j}}})} \nonumber\\ & & \qquad - \sum\limits_{p=1}^r{{(-\rho_0)}^p(N^{(r,p)}_0A^{2(r-p)+3}{A^*}^{r+1}-N^{(r,p)}_1A^{2(r-p)+2}{A^*}^rAA^*+N^{(r,p)}_2A^{2(r-p)+1}{A^*}^rA^2A^*)} \nonumber\\ & & \qquad + \sum\limits_{p=1}^r{{(-\rho_0)}^p\sum\limits_{i=1}^{r-p+1}{N^{(r,p)}_{2i+1}A^{2(r-p-i)+2}{A^*}^r(\sum\limits_{k=1}^i{\sum\limits_{j=0}^2{\rho^{i-k}_0\eta^{(2i+1)}_{k,j}A^{2-j}A^*A^{2k-1+j}}}+\rho_0^i(AA^*-A^*A))}} \nonumber\\& & \qquad - \sum\limits_{p=1}^{r-1}{{(-\rho_0)}^p\sum\limits_{i=1}^{r-p}{N^{(r,p)}_{2i+2}A^{2(r-p-i)+1}{A^*}^r(\sum\limits_{k=0}^i{\sum\limits_{j=0}^2{\rho_0^{i-k}\eta^{(2i+2)}_{k,j}A^{2-j}A^*A^{2k+j}}})}}
\nonumber\\& & \qquad -(-\rho)^{r+1}(N_0^{(r,r+1)}A{A^*}^{r+1}-N_1^{(r,r+1)}{A^*}^rAA^*)\ .\nonumber
\eeqa
The three other monomials $A^{2r+2} {A^*}^{r+1}A$, $A^{2r+1} {A^*}^{r+1}A^2$ and $A^{2r+1}{A^*}^{r+1}$  are also further reduced. For simplicity, corresponding expressions are reported in Appendix B. Now, introduce the combination 
\beqa
\quad &&f_{r+1}(A,A^*)= c^{[r+1,0]}_0 A^{2r+3}{A^*}^{r+1} - c^{[r+1,0]}_1 A^{2r+2}{A^*}^{r+1}A +  c^{[r+1,0]}_2 A^{2r+1}{A^*}^{r+1}A^2 -   \rho_0c^{[r+1,1]}_0 A^{2r+1}{A^*}^{r+1}\    \  \label{combr+1}
\eeqa
with unknown coefficients $c^{[r+1,0]}_j$, $(j=1,2)$, $c^{[r+1,1]}_0$  and normalization $c_0^{[r+1,0]}=1$. Combining all reduced expressions for $A^{2r+3}{A^*}^{r+1}$, $A^{2r+2}{A^*}^{r+1}A$, $A^{2r+1}{A^*}^{r+1}A^2$  and $A^{2r+1}{A^*}^{r+1}$ reported in Appendix B, one observes that  $f_{r+1}(A,A^*)$ generates monomials either of the type (\ref{mongenr}) or  (\ref{mongenUNr}). First, consider monomials of the type (\ref{mongenUNr}) which occur at the lowest order in $\rho_0$, namely $A^{2r}{A^*}^rA^2A^*A$ and $A^{2r}{A^*}^rAA^*A^2$. The conditions under which their coefficients are vanishing read:
\beqa
A^{2r}{A^*}^rA^2A^*A :&&  \qquad  N^{(r,0)}_3\eta^{(3)}_{1,0}+c^{[r+1,0]}_1M_2^{(r,0)}=0\ , \nonumber\\
A^{2r}{A^*}^rAA^*A^2 :&& \qquad N_3^{(r,0)}\eta_{1,1}^{(3)}+c^{[r+1,0]}_2c^{[r,0]}_1=0\ .\nonumber
\eeqa
Using the explicit expressions for $N^{(r,0)}_3$, $\mu^{(3)}_{1,0}$ and $\mu^{(3)}_{1,1}$ given in Appendices A,B, it is easy to solve these these equations. It yields to:
\beqa
c^{[r+1,0]}_1= \left[ \begin{array}{c} 2r+3\\ 1 \end{array}\right]_q  \ ,\qquad c^{[r+1,0]}_2= \left[ \begin{array}{c} 2r+3\\ 2 \end{array}\right]_q \ .\label{coeffr+10}
\eeqa
The conditions under which the coefficients of other unwanted monomials of the type (\ref{mongenUNr}) are vanishing have now to be considered. In particular, similarly to the case $r=2,3$ the coefficients $c^{[r+1,0]}_1$, $c^{[r+1,0]}_2$ arise in  the following set of conditions: 
\ben
&& A^{2r-1}{A^*}^rAA^*A^3: \qquad \quad N_4^{(r,0)}\eta^{(4)}_{1,1}+c^{[r+1,0]}_1M^{(r,0)}_3\eta^{(3)}_{1,1}=0, \nonumber\\
&& A^{2r-1}{A^*}^rA^2A^*A^2: \ \ \quad \quad N^{(r,0)}_4\eta^{(4)}_{1,0}+c^{[r+1,0]}_1M^{(r,0)}_3\eta^{(3)}_{1,0}+c^{[r+1,0]}_2c^{[r,0]}_2=0,\nonumber \\
&& A^{2(r-i)}{A^*}^rAA^*A^{2i+2}: \ \ \  \ N^{(r,0)}_{2i+3}\eta^{(2i+3)}_{i+1,1}+c^{[r+1,0]}_1M^{(r,0)}_{2i+2}\eta^{(2i+2)}_{i,1}+c^{[r+1,0]}_2c^{[r,0]}_{2i+1}\eta^{(2i+1)}_{i,1}=0,\quad i=\overline{1,r},\nonumber \\
&& A^{2(r-i)}{A^*}^rA^2A^*A^{2i+1}: \ \ \ N^{(r,0)}_{2i+3}\eta^{(2i+3)}_{i+1,0}+c^{[r+1,0]}_1M^{(r,0)}_{2i+2}\eta^{(2i+2)}_{i,0}+c^{[r+1,0]}_2c^{[r,0]}_{2i+1}\eta^{(2i+1)}_{i,0}=0,\quad i=\overline{1,r},\nonumber \\
&& A^{2(r-i)+1}{A^*}^rAA^*A^{2i+1}: \ N^{(r,0)}_{2i+2}\eta^{(2i+2)}_{i,1}+c^{[r+1,0]}_1M^{(r,0)}_{2i+1}\eta^{(2i+1)}_{i,1}+c^{[r+1,0]}_2c^{[r,0]}_{2i}\eta^{(2i)}_{i-1,1}=0,\quad i=\overline{2,r},\nonumber\\
&& A^{2(r-i)+1}{A^*}^rA^2A^*A^{2i}: \quad  N^{(r,0)}_{2i+2}\eta^{(2i+2)}_{i,0}+c^{[r+1,0]}_1M^{(r,0)}_{2i+1}\eta^{(2i+1)}_{i,0}+c^{[r+1,0]}_2c^{[r,0]}_{2i}\eta^{(2i)}_{i-1,0}=0\, \quad i=\overline{2,r}.\nonumber
\een
Using the recursion relations in Appendices A,B, we have checked that all above equations are satisfied, as expected.\vspace{1mm}

More generally, one determines all other coefficients $c^{[r+1,0]}_j$ for $j\geq 3$. One finds:
\beqa
c^{[r+1,0]}_3&=& N^{(r,0)}_3\mu^{(3)}_{1,2}={\left[ {\begin{array}{*{20}{c}}
   {2r + 3}  \\
   3  \\
\end{array}} \right]_q},\nonumber\\
c^{[r+1,0]}_4&=& N^{(r,0)}_4\mu^{(4)}_{1,2}+c^{[r+1,0]}_1M^{(r,0)}_3\mu^{(3)}_{1,2}={\left[ {\begin{array}{*{20}{c}}
    {2r + 3}  \\
    4  \\
 \end{array}} \right]_q},\nonumber\\
c^{[r+1,0]}_{2k+1}&=& N^{(r,0)}_{2k+1}\mu^{(2k+1)}_{k,2}+c^{[r+1,0]}_1M^{r,0}_{2k}\mu^{(2k)}_{k-1,2}+c^{[r+1,0]}_2c^{[r,0]}_{2k-1}\mu^{(2k-1)}_{k-1,2},\hspace{.3cm} k=\overline{2,r+1},\nonumber \\
c^{[r+1,0]}_{2k+2}&=& N^{(r,0)}_{2k+2}\mu^{(2k+2)}_{k,2}+c^{[r+1,0]}_1M^{(r,0)}_{2k+1}\mu^{(2k+1)}_{k,2}+c^{[r+1,0]}_2c^{[r,0]}_{2k}\mu^{(2k)}_{k-1,2},\hspace{.3cm} k=\overline{2,r}.\nonumber
\eeqa 
For any $j\geq 0$, one finds that the coefficient  $c^{[r+1,0]}_{j}$ can be simply expressed as a $q-$binomial:
\beqa
c^{[r+1,0]}_j= \left[ \begin{array}{c} 2r+3\\ j \end{array}\right]_q \ .\label{coefbinr}
\eeqa

All coefficients $c^{[r+1,0]}_j$ being obtained, at the lowest order in $\rho_0$ one has  to check that the coefficients of any unwanted term of the type (\ref{mongenUNr}) with $p+s+t= 2r+1,2r-1,...,1$  are systematically vanishing. Using the recursion relations given in Appendices A,B, this has been checked in details.  Then, following the analysis for $r=3$ it remains to determine the coefficient $c^{[r+1,1]}_0$ which contributes at the order $\rho_0$. The condition such that the coefficient of  the monomial $A^{2r} {A^*}^r A A^*$ is vanishing yields to:
\beqa
c^{[r+1,1]}_0={c^{[r,0]}_1}^2-2c^{[r,0]}_2+\frac{c^{[r,0]}_3}{c^{[r,0]}_1}-\frac{c^{[r,1]}_1}{c^{[r,0]}_1}+2c^{[r,1]}_0\ .\label{cr+11}
\eeqa
Using the explicit expression for $c^{[r+1,0]}_j, j=0,1,2$ and $c^{[r+1,1]}_0$, we have checked in details that $f_{r+1}(A,A^*)$ reduces to a combination of monomials of the type (\ref{mongenr}) only.  The reduced expression $f_{r+1}(A,A^*)$  determines uniquely all the remaining coefficients $c^{[r+1,p]}_j$ for $p\geq 1$. For $r$ generic, in addition to (\ref{coefbinr}) and (\ref{cr+11}) one finally obtains:
\beqa
c^{[r+1,r+1]}_0&=& c^{[r+1,1]}_0c^{[r,r]}_0+N^{(r,r+1)}_0,\label{coeff} \\
c^{[r+1,p]}_0&=& N_0^{(r,p)}+c^{[r+1,1]}_0c^{[r,p-1]}_0,\hspace{.3cm} p=\overline{2,r},\nonumber
\eeqa
\vspace{0mm}
\beqa
c_1^{[r+1,1]}&=& N_3^{(r,0)}+c^{[r+1,0]}_1M^{(r,1)}_0,\nonumber \\
c_1^{[r+1,2]}&=& -N_5^{(r,0)}+N_3^{(r,1)}+c^{[r+1,1]}_0c^{[r,0]}_3+c^{[r+1,0]}_1M^{(r,2)}_0,\nonumber \\
c^{[r+1,r+1]}_1&=&\sum_{p=0}^r{(-1)^{r+p}N^{(r,p)}_{2(r-p)+3}}+ c^{[r+1,1]}_0\sum_{p=0}^{r-1}{(-1)^{r+p+1}c^{[r,p]}_{2(r-p)+1}},\nonumber \\
c^{[r+1,p]}_1&=&\sum_{j=0}^{p-1}{(-1)^{j+p+1}N^{(r,j)}_{2(p-j)+1}}+c^{[r+1,0]}_1M^{(r,p)}_0 \nonumber \\& & + c_0^{[r+1,1]}\sum_{j=0}^{p-2}{(-1)^{j+p}c^{[r,j]}_{2(p-j)-1}},\hspace{.3cm} p=\overline{3,r},\nonumber 
\eeqa
\vspace{0mm}
\beqa
c^{[r+1,1]}_2&=&-N^{(r,0)}_4\eta^{(4)}_{0,2}+c^{[r+1,0]}_1M^{(r,0)}_3+c_2^{[r+1,0]}c_0^{[r,1]},\nonumber \\
c^{[r+1,2]}_2&=&N_6^{(r,0)}\eta^{(6)}_{0,2}-N_4^{(r,1)}\eta^{(4)}_{0,2}+c_1^{[r+1,0]}(M_5^{(r,0)}-M_3^{(r,1)})-c^{[r+1,0]}_2c^{[r,2]}_0+c_0^{[r+1,1]}c^{[r,0]}_4\eta^{(4)}_{0,2},\nonumber \\
c^{[r+1,p]}_2&=&\sum_{j=0}^{p-1}{(-1)^{j+p}N^{(r,j)}_{2(p-j)+2}\eta_{0,2}^{(2(p-j)+2)}}+c^{[r+1,0]}_1\sum_{j=0}^{p-1}{(-1)^{j+p+1}M^{(r,j)}_{2(p-j)+1}}\nonumber \\& &+c^{[r+1,0]}_2c^{[r,p]}_0+c_0^{[r+1,1]}\sum_{j=0}^{p-2}{(-1)^{j+p+1}c^{[r,j]}_{2(p-j)}\eta^{(2(p-j))}_{0,2}},\hspace{.3cm} p=\overline{3,r},\nonumber 
\eeqa
\vspace{0mm}
\beqa
c^{[r+1,1]}_3&=&-(N_5^{(r,0)}\eta^{(5)}_{1,2}-N_3^{(r,1)}\eta^{(3)}_{1,2})-c_1^{[r+1,0]}M_4^{(r,0)}\eta^{(4)}_{0,2}+c_0^{[r+1,1]}c_3^{[r,0]}\eta^{(3)}_{1,2}+c_2^{[r+1,0]}c_3^{[r,0]},\nonumber \\
c_3^{[r+1,p]}&=&\sum_{j=0}^p{(-1)^{j+p}N^{(r,j)}_{2(p-j)+3}\eta^{(2(p-j)+3)}_{1,2}} +c_1^{[r+1,0]}\sum_{j=0}^{p-1}{(-1)^{j+p}M^{(r,j)}_{2(p-j)+2}\eta^{(2(p-j)+2)}_{0,2}} 
\nonumber \\& &+ c^{[r+1,0]}_2\sum_{j=0}^{p-1}{(-1)^{j+p+1}c^{[r,j]}_{2(p-j)+1}}+c_0^{[r+1,1]}\sum_{j=0}^{p-1}{(-1)^{j+p+1}c^{[r,j]}_{2(p-j)+1}\eta^{(2(p-j)+1)}_{1,2}},\hspace{.3cm} j=\overline{2,r},\nonumber
\eeqa
\vspace{-0mm}
\beqa
c^{[r+1,1]}_4&=&-(N_6^{(r,0)}\eta^{(6)}_{1,2}-N_4^{(r,1)}\eta^{(4)}_{1,2})-c_1^{[r+1,0]}(M_5^{(r,0)}\eta_{1,2}^{(5)}-M_3^{(r,1)}\eta_{1,2}^{(3)}) -c_2^{[r+1,0]}c_4^{[r,0]}\eta^{(4)}_{0,2}+c_0^{[r+1,1]}c_4^{[r,0]}\eta^{(4)}_{1,2},\nonumber \\
c^{[r+1,p]}_4&=&\sum_{j=0}^{p}{(-1)^{j+p}N^{(r,j)}_{2(p-j)+4}\eta^{(2(p-j)+4)}_{1,2}}+c^{[r+1,0]}_1\sum_{j=0}^p{(-1)^{j+p}M^{(r,j)}_{2(p-j)+3}\eta^{(2(p-j)+3)}_{1,2}}\nonumber \\& &+c^{[r+1,0]}_2\sum_{j=0}^{p-1}{(-1)^{j+p}c^{[r,j]}_{2(p-j)+2}\eta^{(2(p-j)+2)}_{0,2}} +c_0^{[r+1,1]}\sum_{j=0}^{p-1}{(-1)^{j+p+1}c^{[r,j]}_{2(p-j)+2}\eta^{(2(p-j)+2)}_{1,2}},\hspace{.3cm} p=\overline{2,r-1},\nonumber
\eeqa
\vspace{0mm}
\beqa
c^{[r+1,j-k]}_{2k+3}&=&\sum_{p=0}^{j-k}{(-1)^{p+j+k}N^{(r,p)}_{2(j-p)+4}\eta^{(2(j-p)+4)}_{k+1,2}}+c_1^{[r+1,0]}\sum_{p=0}^{j-k}{(-1)^{p+j+k}M^{(r,p)}_{2(j-p)+2}\eta^{(2(j-p)+2)}_{k,2}}\nonumber \\& &+c_2^{[r+1,0]}\sum_{p=0}^{j-k}{(-1)^{p+j+k}c^{[r,p]}_{2(j-p)+1}\eta^{(2(j-p)+1)}_{k,2}}\nonumber \\& &+c_0^{[r+1,1]}\sum_{p=0}^{j-k-1}{(-1)^{p+j+k+1}c^{[r,p]}_{(2(j-p)+1)}\eta^{(2(j-p)+1)}_{k+1,2}},\hspace{.3cm} j=\overline{3,r}, \qquad k=\overline{1,j-2},\nonumber
\eeqa
\beqa
c^{[r+1,1]}_{2j+1}&=&-(N^{(r,0)}_{2j+3}\eta^{(2j+3)}_{j,2}-N^{(r,1)}_{2j+1}\eta^{(2j+1)}_{j,2})-c_1^{[r+1,0]}(\eta^{(2j+2)}_{j-1,2}M^{(r,0)}_{2j+2}-M^{(r,1)}_{2j}\eta^{(2j)}_{j-1,2})\nonumber \\& &-c^{[r+1,0]}_2(c^{[r,0]}_{2j+1}\eta^{(2j+1)}_{j-1,2}-c^{[r,1]}_{2j-1}\eta^{(2j-1)}_{j-1,2})+c_0^{[r+1,1]}c^{[r,0]}_{2j+1}\eta^{(2j+1)}_{j,2},\hspace{.3cm} j=\overline{2,r},\nonumber
\eeqa
\beqa
c^{[r+1,j-k]}_{2k+2}&=&\sum_{p=0}^{j-k}{(-1)^{p+j+k}N^{(r,p)}_{2(j-p)+2}\eta^{(2(j-p)+2)}_{k,2}}+c_1^{[r+1,0]}\sum_{p=0}^{j-k}{(-1)^{p+j+k}M^{(r,p)}_{2(j-p)+1}\eta^{(2(j-p)+1)}_{k,2}}\nonumber \\& & +c_2^{[r+1,0]}\sum_{p=0}^{j-k}{(-1)^{p+j+k}c^{[r,p]}_{2(j-p)}\eta^{(2(j-p))}_{k-1,2}}\nonumber \\& &+c_0^{[r+1,1]}\sum_{p=0}^{j-k-1}{(-1)^{p+j+k+1}c_{(2(j-p))}^{[r,p]}\eta^{(2(j-p))}_{k,2}}, \hspace{.3cm} j=\overline{4,r}, \qquad k=\overline{2,j-2},\nonumber
\eeqa
\beqa
c^{[r+1,1]}_{2j}&=&c^{[r+1,1]}_0c^{[r,0]}_{2j}\eta^{(2j)}_{j-1,2}-N^{(r,0)}_{2j+2}\eta^{(2j+2)}_{j-1,2}+N^{(r,1)}_{2j}\eta^{(2j)}_{j-1,2}\nonumber \\ & &-c^{[r+1,0]}_1(M^{(r,0)}_{2j+1}\eta^{(2j+1)}_{j-1,2}-M^{(r,1)}_{2j-1}\eta^{(2j-1)}_{j-1,2})-c^{[r+1,0]}_2(c^{[r,0]}_{2j}\eta^{(2j)}_{j-2,2}-c^{[r,1]}_{2j-2}\eta^{(2j-2)}_{j-2,2}),\hspace{.3cm} j=\overline{3,r}.\nonumber 
\eeqa
\vspace{4mm}

According to above results and using the automorphism $A\leftrightarrow A^*$ and $\rho_0\leftrightarrow \rho_1$, we conclude that if $A,A^*$ satisfy the defining relations (\ref{qDG1}), (\ref{qDG2}), then the higher order $q-$Dolan-Grady relations (\ref{qDGr}),  (\ref{qDGr2}) are such that the coefficients $c^{[r,p]}_j$ are determined recursively by (\ref{coefbinr}), (\ref{cr+11}) and (\ref{coeff}). For $p\geq1$, they can be computed for practical purpose\footnote{As the reader may have noticed, for $r=2,3$ the coefficients $c_j^{[r,p]}$ are proportional to $[2r+1]_q$ iff $j\neq 0$ or $2r+1$. For a large number of  values $r\geq 4$, this property holds too. As a consequence, the relations (\ref{qDGfinr}) drastically simplify for $q^{2r+1}=\pm1$. This case is however not considered here.}.  In particular, one observes that $c_{j}^{[r,p]}=c_{2(r-p)+1-j}^{[r,p]}$\ . For $r=4,5,...\leq 10$, using a computer program we have checked in details that $r-th$ higher order relations of the form (\ref{qDGfinr}) hold, and that the coefficients satisfy above recursive formula.\vspace{1mm}

\subsection{Comments} Let $A,A^*$ be the fundamental generators of the $q-$Onsager algebra with defining relations (\ref{qDG}). Let $V$ denote an irreducible finite dimensional vector space and assume each of $A,A^*$ is diagonalizable on $V$. Then, it is known \cite[Theorem 3.10]{Ter03} that $A,A^*$ act on $V$ as a TD pair. Assume conjecture \ref{conj1} holds. Then, the coefficients obtained from the two-variable polynomial (\ref{polyqons}) and given by (\ref{cfinr})  must coincide exactly with the coefficients satisfying above recursive formulae. For $r=4,5,...\leq 10$, using a computer program we have checked the correspondence. Also, note that for the special case $\rho_0=\rho_1=0$ Lusztig's higher order $q-$Serre relations for generic values of $r$ are recovered: in this case the coefficients are given by (\ref{coefbinr}) in agreement with  \cite{Luzt}. Besides the proofs for $r=2,3$, these checks give another support for conjecture \ref{conj1}.\vspace{1mm}

To conclude, let us mention that a  proof of conjecture \ref{conj1} for {\it generic values of $r$}  - without using the properties of tridiagonal pairs - would be desirable. In this direction, by analogy with Lusztig's work \cite{Luzt} we expect that the construction of the braid group associated with the $q-$Onsager algebra\footnote{P. Baseilhac and S. Kolb, in progress.} (\ref{qDG}) will help.
\vspace{2mm}

\section{Concluding remarks}
An explicit relationship between the $q-$Onsager algebra (\ref{qDG}) and  a coideal subalgebra of $U_q({\widehat{sl_2}})$ is already known \cite{B1} (see also \cite{IT32},\cite{Kolb}). Recall that there exists an algebra homomorphism from (\ref{qDG}) to $U_q(\widehat{sl_2})$ with scalars $c_i,\overline{c}_i,\epsilon_i \in \K$ such that
\beqa
A&=& c_0e_0q^{h_0/2} + \overline{c}_0f_0q^{h_0/2} + \epsilon_0q^{h_0}\ ,\label{realsl2}\\
 A^*&=& c_1e_1q^{h_1/2} + \overline{c}_1f_1q^{h_1/2} + \epsilon_1q^{h_1} \ ,\nonumber 
\eeqa
where one identifies $\rho_i=c_i\overline{c}_i(q+q^{-1})^2$ for $i=0,1$.  According to conjecture \ref{conj1}, the elements of the coideal subalgebra of $U_q(\widehat{sl_2})$ generated by (\ref{realsl2})  satisfy the higher order $q-$Dolan-Grady relations (\ref{qDGfinr}). If $A,A^*$ act on a finite dimensional vector space $V$ \cite{TD10} - see the explicit examples considered in \cite{BK,BK1} - this is true according to Theorem 2. Also, for any of the special cases $c_i=0$ or  $\overline{c}_i=0$, $\rho_i=0$ so that the relations (\ref{qDGfinr}) reduce to the higher $q-$Serre relations (\ref{hqSerre}). 
\vspace{1mm}

A straightforward application of (\ref{qDGfinr}) concerns the theory of quantum integrable systems, in which case quantum groups provide efficient tools to solve a model. For instance, given a Hamiltonian which commutes with the elements of a quantum group, the structure of the Hamiltonian's spectrum and eigenstates can be studied within the representation theory of the quantum group. For $q$ generic, this approach has been applied to several models. One of the most studied example is the XXZ spin$-s$ chain with periodic or special boundary conditions, and its thermodynamic limit's analogue. For $q$ a root of unity, new features appear\footnote{See e.g. \cite{DFM,RS} and references citing these articles.}. For $s=1/2$ the existence of additional properties  which do not follow from the star-triangle equation has been early noticed by Baxter \cite{Baxter}. Numerically, additional degeneracies have been also observed in the spectrum for arbitrary spin (see references in \cite{DFM,KM})  pointing out the existence of hidden symmetries. A breakthrough was made in \cite{DFM} where it was shown that the model at $q$ a root of unity enjoys a $\widehat{sl_2}$ loop algebra invariance. Remarkably, this property further extends to other integrable  models \cite{KM} (see also \cite{ND,AYP}). For instance, the Fateev-Zamolodchikov XXZ spin chain ($s=1$) \cite{FZ} at $q=e^{i\pi/3}$ which is closely related with the $3-$state superintegrable chiral Potts model. Importantly, in the works \cite{DFM,KM,ND,AYP} the discovery of the hidden $\widehat{sl_2}$ loop symmetry is essentially based on the higher order $q-$Serre relations (\ref{hqSerre}) of $U_q(\widehat{g})$. 
Having this in mind, a natural question is whether other types of loop symmetries can emerge for $q$ a root of unity and certain boundary conditions in open spin chain and related vertex models. Indeed, based on the relation between  a certain coideal subalgebra of $U_q(\widehat{sl_2})$ and the reflection equation associated with the $U_q(\widehat{sl_2})$ $R-$matrix, the spectrum generating algebra associated with the XXZ open spin chain with generic boundary parameters and $q$ has been identified with the $q-$Onsager algebra (\ref{qDG}) \cite{BK,BK1,BK3}. For $q-$root of unity and the boundary parameters tuned appropriately, degeneracies in the spectrum are expected\footnote{See for instance \cite{PS}, \cite{dG}.}. As a consequence, by analogy with the analysis in \cite{DFM,KM,ND,AYP}, the relations (\ref{qDGfinr}) here derived should play a central role in identifying the hidden symmetry of the open XXZ spin chain (as well as other models) at $q$ a root of unity and certain boundary conditions. We intend to study this problem elsewhere.\vspace{1mm}

Finally, let us mention that higher rank generalizations of the $q-$Onsager algebra (\ref{qDG}) have been proposed in \cite{BB1} (see also \cite{Kolb}). By analogy with the case of  $\widehat{g}=\widehat{sl_2}$  here presented,  higher order relations can be constructed following a similar analysis \cite{BV}. We expect these relations will find applications in the theory of tridiagonal algebras associated with higher rank affine Lie algebras $\widehat{g}$.
\vspace{0.5cm}

\noindent{\bf Acknowledgments:} We are indebted to P. Terwilliger for a careful reading of the first version of the manuscript, and sharing with us some of the results presented in Section 2.  P.B thanks S. Baseilhac and S. Kolb for discussions. 
\vspace{0cm}

\vspace{1cm}

\newpage
\centerline{\bf APPENDIX A: Coefficients $ \eta^{(m)}_{k,j}$, $M^{(r,p)}_j$, $N^{(r,p)}_j$}

\vspace{3mm}

The coefficients that appear in eqs. (\ref{eqmu1}), (\ref{eqmu2}) are such that:
\beqa
\eta^{(3)}_{1,0}&=& [3]_q,\qquad \eta^{(3)}_{1,1}=-[3]_q,\qquad \eta^{(3)}_{1,2}=1,\nonumber \\
\eta^{(4)}_{0,0}&=&1,\qquad \eta^{(4)}_{0,1}=q^2+q^{-2},\qquad \eta^{(4)}_{0,2}=-[3]_q,\nonumber \\
\eta^{(4)}_{1,0}&=&(q^2+q^{-2})[3]_q,\qquad \eta^{(4)}_{1,1}=-(q^2+q^{-2})[2]_q^2,\qquad \eta^{(4)}_{1,2}=[3]_q.\nonumber 
\eeqa
The recursion relations for $\eta^{(m)}_{k,j}$ are such that:
\beqa
\eta^{(2n+2)}_{0,0} &=& 1,\nonumber \\
\eta^{(2n+2)}_{k,0}&=&[3]_q\eta^{(2n+1)}_{k,0}+\eta^{(2n+1)}_{k,1},\qquad 1 \leq k \leq n, \nonumber \\
\eta^{(2n+2)}_{0,1}&=&\eta^{(2n+1)}_{1,0}-1, \nonumber \\
\eta^{(2n+2)}_{k,1}&=&-[3]_q\eta^{(2n+1)}_{k,0}+\eta^{(2n+1)}_{k+1,0}+\eta^{(2n+1)}_{k,2},\qquad 1 \leq k \leq n-1, \nonumber \\
\eta^{(2n+2)}_{n,1}&=&-[3]_q\eta^{(2n+1)}_{n,0}+\eta^{(2n+1)}_{n,2},\nonumber \\
\eta^{(2n+2)}_{0,2}&=&-\eta^{(2n+1)}_{1,0}, \nonumber \\
\eta^{(2n+2)}_{k,2}&=&\eta^{(2n+1)}_{k,0}-\eta^{(2n+1)}_{k+1,0},\qquad 1 \leq k \leq n-1, \nonumber \\
\eta^{(2n+2)}_{n,2}&=&\eta^{(2n+1)}_{n,0},\nonumber
\eeqa
and
\beqa
\eta^{(2n+3)}_{k,0}&=&[3]_q\eta^{(2n+2)}_{k-1,0}+\eta^{(2n+2)}_{k-1,1},\qquad 1 \leq k \leq n+1, \nonumber \\
\eta^{(2n+3)}_{k,1}&=&-[3]_q\eta^{(2n+2)}_{k-1,0}+\eta^{(2n+2)}_{k,0}+\eta^{(2n+2)}_{k-1,2},\qquad 1 \leq k \leq n ,\nonumber \\
\eta^{(2n+3)}_{n+1,1}&=&-[3]_q\eta^{(2n+2)}_{n,0}+\eta^{(2n+2)}_{n,2}, \nonumber \\
\eta^{(2n+3)}_{k,2}&=& \eta^{(2n+2)}_{k-1,0}-\eta^{(2n+2)}_{k,0},\qquad 1 \leq k \leq n, \nonumber \\
\eta^{(2n+3)}_{n+1,2}&=& \eta^{(2n+2)}_{n,0}. \nonumber
\eeqa

\vspace{3mm}

The coefficients that appear in eqs. (\ref{eq2}), (\ref{eq3}) are such that:

\beqa
M^{(r,0)}_j&=& c^{[r,0]}_j-c^{[r,0]}_1c^{[r,0]}_{j-1},\qquad j=\overline{2,2r+1}, \nonumber \\
M^{(r,0)}_{2r+2}&=&-c^{[r,0]}_1c^{[r,0]}_{2r+1}, \nonumber \\
M^{(r,p)}_0&=&c^{[r,p]}_0, p=\overline{1,r}, \nonumber \\
M^{(r,p)}_j&=&c^{[r,p]}_j-c^{[r,0]}_1c^{[r,p]}_{j-1},\qquad p=\overline{1,r},\quad j=\overline{1,2(r-p)+1}, \nonumber \\
M^{(r,p)}_{2(r-p)+2}&=&-c^{[r,0]}_1c^{[r,p]}_{2(r-p)+1},\qquad p=\overline{1,r},\nonumber 
\eeqa 
and
\beqa
N^{(r,0)}_j&=&c^{[r,0]}_j-c^{[r,0]}_1c^{[r,0]}_{j-1}+({c^{[r,0]}_1}^2-c^{[r,0]}_2)c^{[r,0]}_{j-2},\qquad j=\overline{3,2r+1},\nonumber \\
N^{(r,0)}_{2r+2}&=&-c^{[r,0]}_1c^{[r,0]}_{2r+1}+({c^{[r,0]}_1}^2-c^{[r,0]}_2)c^{[r,0]}_{2r}, \nonumber \\
N^{(r,0)}_{2r+3}&=&({c^{[r,0]}_1}^2-c^{[r,0]}_2)c^{[r,0]}_{2r+1}, \nonumber \\
N^{(r,1)}_0&=& 0, \nonumber \\
N^{(r,1)}_j&=&({c^{[r,0]}_1}^2-c^{[r,0]}_2)c^{[r,1]}_{j-2}-c^{[r,0]}_1c^{[r,1]}_{j-1}+c^{[r,1]}_j-c^{[r,1]}_0c^{[r,0]}_j,\qquad j=\overline{2,2r-1}, \nonumber \\
N^{(r,1)}_1&=&-c^{[r,0]}_1c^{[r,1]}_0+c^{[r,1]}_1-c^{[r,1]}_0c^{[r,0]}_1,\nonumber\eeqa
\beqa 
N^{(r,1)}_{2r}&=&({c^{[r,0]}_1}^2-c^{[r,0]}_2)c^{[r,1]}_{2r-2}-c^{[r,0]}_1c^{[r,1]}_{2r-1}-c^{[r,1]}_0c^{[r,0]}_{2r}, \nonumber \\
N^{(r,1)}_{2r+1}&=&({c^{[r,0]}_1}^2-c^{[r,0]}_2)c^{[r,1]}_{2r-1}-c^{[r,1]}_0c^{[r,0]}_{2r+1}, \nonumber \\
N^{(r,r+1)}_j&=&-c^{[r,1]}_0c^{[r,r]}_j,\qquad j = \overline{0,1}. \nonumber 
\eeqa
For $2 \leq p \leq r,$ 
\beqa
N^{(r,p)}_j&=&({c^{[r,0]}_1}^2-c^{[r,0]}_2)c^{[r,p]}_{j-2}-c^{[r,0]}_1c^{[r,p]}_{j-1}+c^{[r,p]}_j-c^{[r,1]}_0c^{[r,p-1]}_j, \quad j=\overline{2, 2(r-p)+1}, \nonumber \\
N^{(r,p)}_0&=&c^{[r,p]}_0-c^{[r,1]}_0c^{[r,p-1]}_0,\nonumber \\
N^{(r,p)}_1&=&-c^{[r,0]}_1c^{[r,p]}_0+c^{[r,p]}_1-c^{[r,1]}_0c^{[r,p-1]}_1, \nonumber \\
N^{(r,p)}_{2(r-p)+2}&=&({c^{[r,0]}_1}^2-c^{[r,0]}_2)c^{[r,p]}_{2(r-p)}-c^{[r,0]}_1c^{[r,p]}_{2(r-p)+1}-c^{[r,1]}_0c^{[r,p-1]}_{2(r-p)+2}, \nonumber \\
N^{(r,p)}_{2(r-p)+3}&=&({c^{[r,0]}_1}^2-c^{[r,0]}_2)c^{[r,p]}_{2(r-p)+1}-c^{[r,1]}_0c^{[r,p-1]}_{2(r-p)+3}. \nonumber 
\eeqa

\vspace{1cm}

\newpage

\centerline{\bf APPENDIX B: $A^{2r+2}{A^*}^{r+1}A$, $A^{2r+1}{A^*}^{r+1}A^2$, $A^{2r+1}{A^*}^{r+1}$}

\vspace{3mm}

In addition to (\ref{m1}), the other monomials can be written as:
\begin{eqnarray*}
A^{2r+2}{A^*}^{r+1}A &=&-\sum\limits_{i=1}^r{M^{(r,0)}_{2i+2}A^{2(r-i)}{A^*}^r(\sum\limits_{k=0}^i{\sum\limits_{j=0}^2{\rho_0^{i-k}\eta^{(2i+2)}_{k,j}A^{2-j}A^*A^{2k+j+1}}})} \nonumber\\& & + \sum\limits_{i=1}^r{M^{(r,0)}_{2i+1}A^{2(r-i)+1}{A^*}^r(\sum\limits_{k=1}^i{\sum\limits_{j=0}^2{\rho^{i-k}_0\eta^{(2i+1)}_{k,j}A^{2-j}A^*A^{2k+j}}}+\rho_0^i(AA^*A-A^*A^2))}\nonumber \\ & & - \sum\limits_{p=1}^r{{(-\rho_0)}^p(M^{(r,p)}_0A^{2(r-p)+2}{A^*}^{r+1}A-M^{(r,p)}_1A^{2(r-p)+1}{A^*}^rAA^*A+M^{(r,p)}_2A^{2(r-p)}{A^*}^rA^2A^*A)} \\& & - M^{(r,0)}_2A^{2r}{A^*}^rA^2A^*A-\sum\limits_{p=1}^{r-1}{{(-\rho_0)}^p\sum\limits_{i=1}^{r-p}{M^{(r,p)}_{2i+2}A^{2(r-p-i)}{A^*}^r(\sum\limits_{k=0}^i{\sum\limits_{j=0}^2{\rho_0^{i-k}\eta^{(2i+2)}_{k,j}A^{2-j}A^*A^{2k+j+1}}})}} \nonumber\\& & + \sum\limits_{p=1}^{r-1}{{(-\rho_0)}^p\sum\limits_{i=1}^{r-p}{M^{(r,p)}_{2i+1}A^{2(r-p-i)+1}{A^*}^r(\sum\limits_{k=1}^i{\sum\limits_{j=0}^2{\rho_0^{i-k}\eta^{(2i+1)}_{k,j}A^{2-j}A^*A^{2k+j}}}+\rho_0^i(AA^*A-A^*A^2))}}\ ,
\end{eqnarray*}
\vspace{-0.3cm}
\beqa
\qquad A^{2r+1}{A^*}^{r+1}A^2 &=& c^{[r,0]}_1A^{2r}{A^*}^rAA^*A^2-c^{[r,0]}_2A^{2r-1}{A^*}^rA^2A^*A^2 \nonumber\\& &- \sum\limits_{p=1}^{r-1}{{(-\rho_0)}^p(c^{[r,p]}_0A^{2(r-p)+1}{A^*}^{r+1}A^2-c^{[r,p]}_1A^{2(r-p)}{A^*}^rAA^*A^2+c^{[r,p]}_2A^{2(r-p)-1}{A^*}^rA^2A^*A^2)} \nonumber\\& &-{(-\rho_0)}^r(c^{[r,r]}_0A{A^*}^{r+1}A^2-c^{[r,r]}_1{A^*}^rAA^*A^2) \nonumber\\& &+\sum\limits_{i=1}^r{c^{[r,0]}_{2i+1}A^{2(r-i)}{A^*}^r(\sum\limits_{k=1}^i{\sum\limits_{j=0}^2{\rho_0^{i-k}\eta^{(2i+1)}_{k,j}A^{2-j}A^*A^{2k+1+j}}}+\rho_0^i(AA^*A^2-A^*A^3))} \nonumber\\ & & - \sum\limits_{i=1}^{r-1}{c^{[r,0]}_{2i+2}A^{2(r-i)-1}{A^*}^r(\sum\limits_{k=0}^i{\sum\limits_{j=0}^2{\rho_0^{i-k}\eta^{(2i+2)}_{k,j}A^{2-j}A^*A^{2k+j+2}}})}\nonumber \\& &  +\sum\limits_{p=1}^{r-1}{{(-\rho_0)}^p\sum\limits_{i=1}^{r-p}{c^{[r,p]}_{2i+1}A^{2(r-p-i)}{A^*}^r(\sum\limits_{k=1}^i{\sum\limits_{j=0}^2{\rho_0^{i-k}\eta^{(2i+1)}_{k,j}A^{2-j}A^*A^{2k+1+j}}}+\rho_0^i(AA^*A^2-A^*A^3))}} \nonumber\\& & -\sum\limits_{p=1}^{r-2}{{(-\rho_0)}^p\sum\limits_{i=1}^{r-p-1}{c^{[r,p]}_{2i+2}A^{2(r-p-i)-1}{A^*}^r(\sum\limits_{k=0}^i{\sum\limits_{j=0}^2{\rho_0^{i-k}\eta^{(2i+2)}_{k,j}A^{2-j}A^*A^{2k+j+2}}})}}\ ,\nonumber
\eeqa
\vspace{-0.3cm}
\begin{eqnarray*}
A^{2r+1}{A^*}^{r+1}&=&c^{[r,0]}_1A^{2r}{A^*}^rAA^*-c^{[r,0]}_2A^{2r-1}{A^*}^rA^2A^* \\& &+\sum\limits_{k=1}^r{c^{[r,0]}_{2k+1}A^{2(r-k)}{A^*}^r(\sum\limits_{i=1}^k{\sum\limits_{j=0}^2{\mu^{(2k+1)}_{i,j}A^{2-j}A^*A^{2i-1+j}}}+\rho^k(AA^*-A^*A))} \\& & -\sum\limits_{k=1}^{r-1}{c^{[r,0]}_{2k+2}A^{2(r-k)-1}{A^*}^r(\sum\limits_{i=0}^k{\sum\limits_{j=0}^2{\mu^{(2k+2)}_{i,j}A^{2-j}A^*A^{2i+j}}})} \\& & -\rho^r(c^{[r,r]}_0A{A^*}^{r+1}-c^{[r,r]}_1{A^*}^rAA^*) \\& & -\sum\limits_{p=1}^{r-1}{\rho^p(c^{[r,p]}_0A^{2(r-p)+1}{A^*}^{r+1}-c^{[r,p]}_1A^{2(r-p)}{A^*}^rAA^*+c^{[r,p]}_2A^{2(r-p)-1}{A^*}^rA^2A^*)} \\& & + \sum\limits_{p=1}^{r-1}{\rho^p\sum\limits_{k=1}^{r-p}{c^{[r,p]}_{2k+1}A^{2(r-p-k)}{A^*}^r(\sum\limits_{i=1}^k{\sum\limits_{j=0}^2{\mu^{(2k+1)}_{i,j}A^{2-j}A^*A^{2i-1+j}}}+\rho^k(AA^*-A^*A))}} \\& & - \sum\limits_{p=1}^{r-2}{\rho^p\sum\limits_{k=1}^{r-p-1}{c^{[r,p]}_{2k+2}A^{2(r-p-k)-1}{A^*}^r(\sum\limits_{i=0}^k{\sum\limits_{j=0}^2{\mu^{(2k+2)}_{i,j}A^{2-j}A^*A^{2i+j}}})}}.
\end{eqnarray*}


\newpage



\begin{thebibliography}{10}


\bibitem[AYP]{AYP}
H. Au-Yang and J.H.H. Perk,  2011 {\it Quantum loop subalgebra and eigen
vectors of the superintegrable chiral Potts transfer matrices}, J. Phys. {\bf A 44} 025205, {\tt arXiv:0907.0362};\\
H. Au-Yang and J.H.H. Perk, {\it Serre Relations in the Superintegrable Model},  {\tt arXiv:1210.5803}.
%
\bibitem[Bas]{B1}
P. Baseilhac, {\it Deformed Dolan-Grady relations in quantum integrable models}, Nucl.Phys. {\bf B 709} (2005) 491-521, {\tt arXiv:hep-th/0404149};\\
P. Baseilhac, {\it An integrable structure related with tridiagonal algebras}, Nucl.Phys. {\bf B 705} (2005) 605-619, {\tt arXiv:math-ph/0408025}. 
%
%
\bibitem[Bax]{Baxter}
R. J. Baxter, Ann. Phys. {\bf 76} (1973) 1; {\bf 76} (1973) 25; {\bf 76} (1973) 48.
%
\bibitem[BB1]{BB0}
 P. Baseilhac and S. Belliard,    {\it Central extension of the reflection equations and an analog of Miki's formula},  J. Phys. A {\bf 44} (2011) 415205, {\tt arXiv:1104.1591}. 
%
\bibitem[BB2]{BB1}
P. Baseilhac and S. Belliard, {\it Generalized q-Onsager algebras and boundary affine Toda field theories}, Lett. Math. Phys. {\bf 93} (2010) 213-228, {\tt arXiv:0906.1215}. 
%

\bibitem[BK1]{BK}
P. Baseilhac and K. Koizumi, {\it A new (in)finite dimensional algebra for quantum integrable models}, Nucl. Phys. {\bf B 720} (2005) 325--347, {\tt arXiv:math-ph/0503036}.
%
\bibitem[BK2]{BK1}
P. Baseilhac and K. Koizumi, {\it A deformed analogue of Onsager's symmetry in the XXZ open spin chain}, J.Stat.Mech. {\bf 0510} (2005) P005, {\tt arXiv:hep-th/0507053}.
%
\bibitem[BK3]{BK3}
P. Baseilhac and K. Koizumi, {\it Exact spectrum of the XXZ open spin chain from the q-Onsager algebra representation theory}, J. Stat. Mech.  (2007) P09006, {\tt arXiv:hep-th/0703106}.
%
\
%
\bibitem[BS1]{BasS}
P. Baseilhac and K. Shigechi, {\it A new current algebra and the reflection equation}, Lett. Math. Phys. {\bf  92} (2010) 47-65, {\tt arXiv:0906.1215}.
%

\bibitem[BV]{BV}
P. Baseilhac and T.T. Vu,  {\it Higher order relations for ADE-type
generalized $q-$Onsager algebras}, to appear.
%
\bibitem[Cher]{Cher84}
I.V. Cherednik, {\it Factorizing particles on the half-line and root systems}, Teor. Mat. Fiz. {\bf 61} (1984) 35--44.
%


\bibitem[DFM]{DFM}
 T. Deguchi, K. Fabricius and B. M. McCoy, {\it The $sl_2$ loop algebra symmetry of the six-vertex model at roots of unity}, J. Statist. Phys. {\bf 102} (2001) 701--736,  {\tt arXiv:cond-mat/9912141}.

%
\bibitem[DG]{DG}
L. Dolan and M. Grady, {\it Conserved charges from self-duality}, Phys. Rev. {\bf D 25} (1982) 1587--1604.
%


\bibitem[Dr]{Dr}
V.G. Drinfeld, {\it Quantum groups}, Proceedings ICM 1986, Amer. Math. Soc., 1987,
pp. 798--820.

\bibitem[FZ]{FZ}
V. Fateev and A. Zamolodchikov, Sov. J. Nucl. Phys. {\bf 32} (1980) 581.

\bibitem[GNPR]{dG}
J. de Gier, A. Nichols, P. Pyatov, and V. Rittenberg, {\it Magic in the spectra of the XXZ quantum chain with
boundaries at $\Delta = 0$ and $\Delta = 1/2$}, Nucl. Phys. {\bf B
729} (2005) 387, {\tt arXiv:hep-th/0505062v2}.
%

\bibitem[J]{Jim} 
M. Jimbo, {\it A $q-$analogue of $U(g)$ and the Yang-Baxter equation}, Lett. Math. Phys. {\bf 11}
(1985) 63--69.


\bibitem[K]{Kolb}
S. Kolb,  {\it Quantum symmetric Kac-Moody pairs}, {\tt arXiv:1207.6036v1}.

\bibitem[L]{Luzt}
G. Lusztig, {\it Introduction to Quantum Groups}, Birkhauser (1993).



\bibitem[INT]{TD10}
T.~Ito, K.~Nomura, and P.~Terwilliger, {\it A classification of sharp tridiagonal pairs}, Linear Algebra Appl. {\bf 435} (2011)1857--1884, {\tt arXiv:1001.1812v1}.

%
\bibitem[IT1]{IT03}
T. Ito and P. Terwilliger, {\it The shape of a tridiagonal pair}, J. Pure Appl. Algebra {\bf 188} (2004) 145--160, {\tt arXiv:math.QA/0304244v1}.
%

\bibitem[IT2]{IT32}
T. Ito and P. Terwilliger,   {\it  Tridiagonal pairs and the quantum affine algebra $U_q({\hat {sl}}_2)$}, Ramanujan J. {\bf 
13} (2007) 39--62, {\tt arXiv:math/0310042}.
%


\bibitem[ITT]{TD00}
T.~Ito, K.~Tanabe and P.~Terwilliger, {\it Some algebra related to ${P}$- and ${Q}$-polynomial association
  schemes},  in:
Codes and Association Schemes (Piscataway NJ, 1999), Amer. Math. Soc., Providence RI, 2001, pp 167--192; 
{\tt arXiv:math.CO/0406556}.


\bibitem[KM]{KM}
C. Korff and B. M. McCoy, {\it Loop symmetry of integrable vertex models at roots of unity}, Nucl. Phys. {\bf B618} (2001) 551--569 
{\tt arXiv:hep-th/0104120}. 

\bibitem[ND]{ND}
 A. Nishino and T. Deguchi, {\it   The $L(sl_2)$ symmetry of the Bazhanov-Stroganov model associated with the superintegrable chiral Potts model}, Phys. Lett. {\bf A 356} (2006)  366--70, {\tt arXiv:cond-mat/0605551};\\
A. Nishino and T. Deguchi, , {\it An algebraic derivation of the eigens
paces associated with an Ising-like spectrum of the superintegrable chiral Potts model}, J. Stat. Phys.
{\bf 133} (2008) 587--615, {\tt arXiv:0806.1268}.
%
\bibitem[NT]{NT:muqrac}
K.~Nomura and P.~Terwilliger, {\it Tridiagonal pairs of $q$-Racah type and the $\mu$-conjecture}, Linear Algebra Appl. {\bf 432} (2010) 320--3209, {\tt arXiv:0908.3151}.
%
\bibitem[PS]{PS}
V. Pasquier and H. Saleur, {\it Common structures between finite systems and conformal field theories through quantum groups}, Nucl. Phys. {\bf B 330} (1990) 523--556.
%
\bibitem[RS]{RS}
Yu. G. Stroganov,  {\it The importance of being odd}, J. Phys. {\bf A 34} (2001) L179--L185, {\tt  arXiv:cond-mat/0012035};\\
A.V. Razumov and Yu. G. Stroganov, {\it Spin chains and combinatorics}, J. Phys. {\bf A 34} (2001) 3185-3190, {\tt  arXiv:cond-mat/0012141}.
%
\bibitem[Sk]{Skly88}
E.K. Sklyanin, {\it Boundary conditions for integrable quantum systems}, J. Phys. {\bf A 21} (1988) 2375--2389.
%
\bibitem[Ter]{Ter03}
P. Terwilliger, {\it Two relations that generalize the $q-$Serre
relations and the Dolan-Grady relations}, Proceedings of the Nagoya 1999 International workshop on physics and combinatorics. Editors A. N. Kirillov, A. Tsuchiya, H. Umemura. pp 377--398, {\tt arXiv:math.QA/0307016}.
%



\end{thebibliography}
\end{document}